\documentclass[article]{IEEEtran}

\usepackage{amsfonts}
\usepackage{amssymb}
\usepackage{amsthm}
\usepackage{times,color}
\usepackage[cmex10]{amsmath}
\usepackage{graphicx}
\usepackage{float}
\usepackage{xcolor}
\usepackage{graphicx}
\usepackage{times,color}
\usepackage{amsfonts}
\usepackage[draft,ulem=normalem]{changes}
\newtheorem{theorem}{Theorem}
\newtheorem{Construction}{Construction}

\newtheorem{lemma}{Lemma}
\newtheorem{corollary}{Corollary}

\newtheorem{definition}{Definition}

\newcommand{\nchoosek}[2]{\left(\begin{array}{c}#1\\#2\end{array}\right)}
\newcommand{\comment}[1]{}

\newtheorem{exam}{Example$\!$}
\newenvironment{example}{\begin{exam}\hspace*{-1ex}{\bf }}{\end{exam}}
\newenvironment{proof-sketch}{\noindent \textit{Sketch of Proof:}}{$\blacksquare$}

\newtheorem{remrk}{Remark$\!$}
\newenvironment{remark}{\begin{remrk}\hspace*{-1ex}{\bf }}{\end{remrk}}

\newtheorem*{rep@theorem}{\rep@title}
\newcommand{\newreptheorem}[2]{%
\newenvironment{rep#1}[1]{%
 \def\rep@title{#2 \ref{##1}}%
 \begin{rep@theorem}}%
 {\end{rep@theorem}}}
\makeatother

\newreptheorem{theorem}{Theorem}
\newreptheorem{lemma}{Lemma}
\newtheorem{claim}{Claim}
\newreptheorem{claim}{Claim}
\newreptheorem{Construction}{Construction}

\newcommand{\cA}{{\cal A}}
\newcommand{\cB}{{\cal B}}
\newcommand{\cC}{{\cal C}}

\newcommand{\cE}{{\cal E}}
\newcommand{\cG}{{\cal G}}
\newcommand{\cH}{{\cal H}}
\newcommand{\cI}{{\cal I}}
\newcommand{\cN}{{\cal N}}

\newcommand{\cS}{{\cal S}}
\newcommand{\cT}{{\cal T}}
\newcommand{\cX}{{\cal X}}

\newcommand{\cW}{{\cal W}}

\newcommand{\bfa}{{\boldsymbol a}}
\newcommand{\bfb}{{\boldsymbol b}}
\newcommand{\bfc}{{\boldsymbol c}}

\newcommand{\bfe}{{\boldsymbol e}}

\newcommand{\bfg}{{\boldsymbol g}}
\newcommand{\bfh}{{\boldsymbol h}}

\newcommand{\bfk}{{\boldsymbol k}}

\newcommand{\bfp}{{\boldsymbol p}}

\newcommand{\bft}{{\boldsymbol t}}
\newcommand{\bfu}{{\boldsymbol u}}
\newcommand{\bfv}{{\boldsymbol v}}

\newcommand{\bfx}{{\boldsymbol x}}
\newcommand{\bfy}{{\boldsymbol y}}
\newcommand{\bfz}{{\boldsymbol z}}

\begin{document}

\title{Correcting Grain-Errors in Magnetic Media \thanks{The paper was presented in part at IEEE ISIT 2013.}}

\author{
  \IEEEauthorblockN{
    Ryan~Gabrys\IEEEauthorrefmark{1},~
    Eitan~Yaakobi\IEEEauthorrefmark{2},~and
    Lara~Dolecek\IEEEauthorrefmark{1}}
  {\normalsize
    \begin{tabular}{cc}
      \IEEEauthorrefmark{1}University of California, Los Angeles~~ &
      \IEEEauthorrefmark{2}California Institute of Technology\\
           rgabrys@ucla.edu, dolecek@ee.ucla.edu & yaakobi@caltech.edu \\
    \end{tabular}}
    }

\maketitle

\begin{abstract} 
This paper studies new bounds and constructions that are applicable to the combinatorial granular channel model previously introduced by Sharov and Roth. We derive new bounds on the maximum cardinality of a grain-error-correcting code and propose constructions of codes that correct grain-errors.  We demonstrate that a permutation of the classical group codes (e.g., Constantin-Rao codes) can correct a single grain-error.  In many cases of interest, our results improve upon the currently best known bounds and constructions. Some of the approaches adopted in the context of grain-errors may have application to related channel models. \end{abstract}

\section{Introduction}
Granular media is a promising magnetic recording technology that currently presents formidable challenges to achieving capacity. One of the main issues with granular media is the uncertainty of the locations of the grains in the underlying recording medium. Typically, this medium is organized into grains whose locations and sizes are random. Information is stored by controlling the magnetization of the individual grains so that each grain can store a single bit of data \cite{White}, \cite{Wood}.

The read and write processes are typically unaware of the locations of the grains. As a result, the medium is divided into evenly spaced bit cells and the information is written into these bit cells \cite{arya}. In the traditional setup, the bit cell is usually larger than a single-grain. When the size of the bit cells is reduced enough, the effects of the random positions of the grains become pronounced. In particular, in \cite{Wood} a one-dimensional channel model was studied that illustrated the effects of having grains with randomly selected lengths of $1$, $2$, or $3$ bits. When grains span more than a single bit cell, the polarity of a grain is set by the last bit written into it. The errors manifest themselves as overwrites (or \textit{\textbf{smears}}) where the last bit in the grain overwrites the rest of the bits in the grain. In this work, the focus is on grains of length one or two bits. A \textit{\textbf{grain-error}} is an error where the information from one bit overwrites the information stored in the preceding bit in the grain. Without loss of generality, and as in \cite{arya}, our model assumes that the first bit smears the following adjacent bit in the grain.


In \cite{artyom}, Sharov and Roth presented combinatorial bounds and code constructions for granular media. In \cite{ari}, Iyengar, Siegel, and Wolf studied a related model from an information-theoretic perspective. In \cite{arya}, Mazumdar, Barg, and Kashyap introduced  a channel model and studied coding methods for a one-dimensional granular magnetic medium. In \cite{arya}, the focus was on binary alphabets and the types of errors studied in \cite{arya} will be referred in this work as \textit{\textbf{non-overlapping grain-errors}}. In \cite{artyom}, Sharov and Roth generalized the model and considered non-binary alphabets as well as \textit{\textbf{overlapping grain-errors}}. Overlapping grain-errors permit the occurrence of two errors in consecutive positions whereas non-overlapping grain-errors cannot be adjacent. Note that there is no distinction between a non-overlapping single grain-error and an overlapping single grain-error. In this work, we restrict our attention to the overlapping grain-error model where, similar to the works of \cite{arya} and \cite{artyom}, we consider only grain-errors of length two. The overlapping grain-error model was chosen because a) overlapping grain errors are common in bit-patterned media recording \cite{ari} and b) codes that correct overlapping grain-errors can be used to correct non-overlapping grain-errors. As will be discussed later, we briefly note that since the set of possible error patterns under the non-overlapping grain-error model is a subset of the set of possible error patterns under the overlapping grain-error model, an upper bound on the overlapping grain-error model is not an upper bound on the non-overlapping grain-error model. We say that a code is a \textit{\textbf{$\bft$-grain-error-correcting code}} if it can correct up to $t$ overlapping grain-errors. In both \cite{arya} and \cite{artyom}, bounds and constructions were given. Recently, in \cite{Kashyap} some of the techniques from \cite{Negar} were adopted to obtain improved upper bounds on the maximum cardinalities of non-overlapping grain-error codes.

The main contribution of this paper is to offer new cardinality bounds of codes correcting grain-errors and to provide new explicit constructions of such codes.
We show that the class of group codes from \cite{cr} is a special case of our general code construction. 
In addition, and similar to \cite{Kashyap}, we provide non-asymptotic upper bounds on the cardinalities of $t$-grain-error-correcting codes, with an explicit expression for the cases where $t=1,2,3$. We show that in many cases our bounds and constructions improve upon the state of the art results from  \cite{arya} and \cite{artyom}. 

Section~\ref{sec:defs} formally defines the channel model  and introduces the notation and tools used for the remainder of the paper. Section~\ref{sec:bound} improves upon the existing upper bounds from \cite{artyom}. Section~\ref{sec:constructions} contains constructions for codes that correct grain-errors and a related type of error which we refer to as mineral-errors. Lower bounds on the cardinalities for some of these codes are then derived in Section~\ref{sec:cardinalities}. Section~\ref{sec:improvedsingle} revisits the general approach to correcting grain/mineral-errors from Section~\ref{subsec:improvedgrains}, and identifies additional codes for certain code lengths. Section~\ref{sec:conclude} concludes the paper. Preliminary results of this work are presented in \cite{gabrys}.

\section{Preliminaries}\label{sec:defs}

In this section, we describe in detail the structure of grain-errors. Afterwards, we introduce some key notation. Section~\ref{subsec:errors} introduces the errors of interest. Section~\ref{subsec:toolsub} reviews the tools which will be used for computing upper bounds. Section~\ref{sec:confusabilitygraph} briefly introduces some graph notation. Section~\ref{subsec:distances} reviews some distance metrics and group codes that will be useful for constructing grain-error-correcting codes. Finally, Section~\ref{subsec:fourier} includes some Fourier analysis tools useful for computing lower bounds for grain-error codes.

\subsection{Grain-errors and mineral-errors}\label{subsec:errors}
In this subsection, we formally introduce the notation and the errors of interest that will be studied in this work. We consider the case where each grain contains either one or two bits of data. A grain-error causes the two bits in the same two-bit grain to either both be $0$ or both be $1$; the error operation can be interpreted as a \textit{\textbf{smearing}}. Following the setup of \cite{arya}, we assume that the first bit smears the second. The problem of interest is how to correct grain-errors when the locations and lengths of the grains are unknown to both the encoder and decoder.

Before continuing, we provide a formal definition of a $t$-grain-error. For a vector $\bfx \in GF(2)^n$, $wt(\bfx)$ refers to the Hamming weight of $\bfx$ and $supp(\bfx)$ denotes the set of indices of $\bfx$ with non-zero values.

\begin{definition}\label{def:grainerror} Let $t \geq 1$ be an integer. Suppose a vector $\bfx \in GF(2)^n$ was stored. Let $\bfe_{\bfx}=(e_1, \ldots, e_n) \in GF(2)^n$, and suppose the vector $\bfy=\bfx + \bfe_{\bfx}$ was read. Then, we say that $\bfe_{\bfx}$ is a \textit{\bf{$\bft$-grain-error for}} $\bfx$ if the following holds:
\begin{enumerate}
\item $wt(\bfe_{\bfx}) \leq t$ and $e_{1}=0$,
\item For $2 \leq i \leq n$, if $e_i \neq 0$, then $x_{i} \neq x_{i-1}$.
\end{enumerate}
\end{definition}
Note that $\bfe_{\bfx}$ depends on the input vector $\bfx$. For shorthand, we say that $\bfe_{\bfx}$ is a $t$-grain-error if the vector $\bfx$ is clear from the context. Notice in Definition~\ref{def:grainerror} that an error at position $i$ where $2 \leq i \leq n$ can be interpreted as a smearing where the value of $\bfx$ at position $i-1$ smears the value of $\bfx$ at position $i$.

A code that can correct all $t$-grain-errors (of length $n$) will be referred to as a \textit{\textbf{$\bft$-grain-error-correcting code}} (of length $n$). For shorthand, a code that can correct a single grain-error will also be referred to as a \textit{\textbf{single-grain code}}. More generally, codes that correct a prescribed number of grain-errors are called \textit{\textbf{grain codes}}. The maximum size of a $t$-grain-error-correcting-code of length $n$ will be referred to as $M(n,t)$. 

Definition~\ref{def:grainerror} coincides with the \textit{\textbf{overlapping grain-error}} model discussed in \cite{artyom}.   We briefly note that since the set of possible errors under the original model of \textit{\textbf{non-overlapping grain-errors}} \cite{arya} is contained within the set of possible errors under the overlapping grain-error model, the code constructions in this paper apply to both models.
We compare the upper bounds derived in Section~\ref{sec:bound} against existing bounds for the overlapping grain-error model (\cite{artyom}). For the remainder of the paper, the term grain-error refers to an overlapping grain-error as stated in Definition~\ref{def:grainerror}.

Suppose a vector $\bfx \in GF(2)^n$ is stored. Let $\cB_{t,G}(\bfx)$ be the set of all possible vectors received (the \textit{\textbf{error-ball}}) given that any $t$-grain-error may occur in $\bfx$. That is, we define $$\cB_{t,G}(\bfx) = \{ \bfx + \bfe_{\bfx} | \text{ $\bfe_{\bfx}$ is a $t$-grain-error} \},$$
and $b_{t,n}(\bfx) = | \cB_{t,G}(\bfx) |$. 
The subscript $G$ refers to grain-errors.

\begin{example}\label{ex:grainerr} Suppose $\bfx=(0, 0, 0, 1, 0)$ was stored. Then, $\cB_{{1,G}}(\bfx) = \{ (0, 0, 0, 1, 0), (0, 0, 0, 1, 1), (0, 0, 0, 0, 0) \}$ and $b_{1,5}(\bfx) = 3$. Notice also that $\cB_{{2,G}}(\bfx) = \{ (0, 0, 0, 1, 0), (0, 0, 0, 1, 1), (0, 0, 0, 0, 0), (0, 0, 0, 0, 1) \}$ and $b_{2,5}(\bfx) = 4$.
\end{example}

We note that the last vector, $(0, 0, 0, 0, 1)$, enumerated in $\cB_{{2,G}}(\bfx)$ for Example~\ref{ex:grainerr} was an overlapping grain-error in the sense that the grain-errors were adjacent so that the bit in position $4$ is both smeared and smearing. 

We introduce a new type of error that will be useful in subsequent analysis.

\begin{definition}\label{def:mineralerror} Let $t \geq 1$ be an integer. Suppose a vector $\bfx \in GF(2)^n$ was stored. Let $\bfe_{\bfx} = (e_1, \ldots, e_n) \in GF(2)^n$ and suppose the vector $\bfy=\bfx + \bfe_{\bfx}$ was received. Then, we say that $\bfe_{\bfx}$ is a \textit{\bf{$\bft$-mineral-error for}} $\bfx$ if the following holds:
\begin{enumerate}
\item $wt(\bfe_{\bfx}) \leq t$,
\item For $2 \leq i \leq n$, if $e_i \neq 0$, then $x_{i} \neq x_{i-1}$.
\end{enumerate}
\end{definition}

Similar to the grain-error setup, we say that $\bfe_{\bfx}$ is a $t$-mineral-error if the vector $\bfx$ is clear from the context. A code that can correct all $t$-mineral-errors of (length $n$) will be referred to as a \textit{\textbf{$\bft$-mineral-error-correcting code}} (of length $n$). \textit{\textbf{Single-mineral codes}} and \textit{\textbf{mineral codes}} are defined analogously as grain codes.

For a given vector $\bfx \in GF(2)^n$, let $\cB_{t,M}(\bfx)$ denote the error-ball for $\bfx$ given that any $t$-mineral-error may occur in $\bfx$. That is, we define $$\cB_{t,M}(\bfx) = \{ \bfx + \bfe_{\bfx} | \text{ $\bfe_{\bfx}$ is a $t$-mineral-error} \}.$$ 

The subscript $M$ refers to mineral-errors.

A useful consequence of Definition~\ref{def:mineralerror} is stated in the following claim.

\begin{claim}\label{cl:mingrain} Suppose $\cC$ is a $t$-grain-error-correcting code. Then, for any two distinct codewords $\bfx=(x_1, \ldots, x_n), \bfy=(y_1, \ldots, y_n) \in \cC$, either
\begin{enumerate}
\item $x_1 \neq y_1$, or
\item $\cB_{t,M}(x_2, \ldots, x_{n}) \cap \cB_{t,M}(y_2, \ldots, y_{n}) = \emptyset$. 
\end{enumerate}
\end{claim}

Suppose $\bfx \in GF(2)^n$ and $\cB_{t,U}$ denotes the error-ball for $\bft$ \textbf{\textit{unrestricted-errors}} where $t$ unrestricted-errors are defined as any binary vector of length $n$ with weight at most $t$. Then, for any vector $\bfx \in GF(2)^n$, $|\cB_{t,U}(\bfx)| = \sum_{i=0}^t \nchoosek{n}{i}$. 

The following lemma follows from the definitions of grain-errors and mineral-errors.

\begin{claim}\label{cl:ordering} For any vector $\bfx \in GF(2)^n$, $\cB_{t,G}(\bfx) \subseteq \cB_{t,M}(\bfx) \subseteq \cB_{t,U}(\bfx)$. \end{claim}

We now present some simple results that follow from the structure of grain-errors. Lemmas~\ref{lem:runs}, \ref{lem:runsdec}, and \ref{lem:even} will be used in Section~\ref{sec:bound} for obtaining upper bounds on the cardinality of grain codes and Lemma~\ref{lem:inf} and Claim~\ref{cl:gmcodes} will be used for constructing grain codes in Section~\ref{sec:constructions}.

A $\textit{\textbf{run}}$ is a maximal substring of one or more consecutive identical symbols. We denote the number of runs in a vector $\bfx$ as $r(\bfx)$ where $\bfx \in GF(2)^n$.

\begin{lemma}\label{lem:runs} For any vector $\bfx$, $b_{t,n}(\bfx)  = \sum_{j=0}^{\min\{t,r(\bfx)-1\}} \nchoosek{r(\bfx)-1}{j}$. \end{lemma}
\begin{proof} Suppose a vector $\bfx$ was stored and that it consists of $k = r(\bfx)$ runs. By Definition~1, a grain-error can occur only at the boundaries between runs. If there are exactly $k \geq t+1$ runs, there are $k-1$ transitions between runs and therefore $b_{t,n}(\bfx) = \sum_{j=0}^{t} \nchoosek{k-1}{j}$. If there are $t$ or fewer runs (i.e., $k \leq t$), then $b_{t,n}(\bfx) = \sum_{j=0}^{k-1} \nchoosek{k-1}{j}$.
\end{proof}

The following lemma is a consequence of the smearing effect of a grain-error. Let the map $\Psi : GF(2)^s \to GF(2)^{s-1}$ be defined so that $\Psi(\bfz) = \bfz' = (z_1', \ldots, z_{s-1}')$ where $z_i' = (z_i + z_{i+1})  \mod 2$ (for $ 1 \leq i \leq s-1$). Notice that $\Psi(\bfz)$ is a linear map and it has a $1$ in position $i$ if and only if $z_{i} \neq z_{i+1}$. Recall that $supp(\bfz)$ refers to the set of non-zero indices in $\bfz$ and $wt(\bfz)$ refers to the Hamming weight of $\bfz$.

\begin{lemma}\label{lem:runsdec} For any two vectors $\bfx, \bfy \in GF(2)^n$ if $\bfy \in \cB_{{t,G}}(\bfx)$, then $ r(\bfy) \leq r(\bfx)$ and $b_{t,n}(\bfy) \leq b_{t,n} (\bfx).$  \end{lemma}
%
%
%
\begin{proof} For the result to hold, we need to show that for any two vectors $\bfx, \bfy \in GF(2)^n$ where $\bfy \in \cB_{{t,G}}(\bfx)$, $r(\bfy) \leq r(\bfx)$. If $r(\bfy) \leq r(\bfx)$, then from Lemma~\ref{lem:runs}, $b_{t,n}(\bfy) \leq b_{t,n} (\bfx).$ Equivalently, we will show that $wt(\Psi(\bfy)) \leq wt(\Psi(\bfx))$. Since $\bfy \in \cB_{t,G}(\bfx)$ we can write $\bfy = \bfx + \bfe_{\bfx}$ where $\bfe_{\bfx}$ is a $t$-grain-error. Let $\bfx' = \Psi(\bfx), \bfe' = \Psi(\bfe_{\bfx}), \bfy' = \Psi(\bfy)$. By the linearity of the map $\Psi$, we can write $\bfy' = \bfx' + \bfe'$ and so $wt(\bfy') = wt(\bfx') + wt(\bfe') - 2 |supp(\bfx') \cap supp(\bfe')|$. In the following, we show $wt(\bfy') \leq wt(\bfx')$ by proving $|supp(\bfx') \cap supp(\bfe')| \geq \frac{wt(\bfe')}{2}$. The proof will follow by induction on the number of runs of $1$s in $\bfe_{\bfx}$. 

We first prove that for any $t$-grain-error $\bfe_{\bfx}$ of length $n$, if $\bfe_{\bfx}$ has a single run of $1$s, then $r(\bfy) \leq r(\bfx)$. Suppose then that $\bfe_{\bfx}=(e_1, \ldots, e_n)$ is a $t$-grain-error and that $\bfe_{\bfx}$ contains a single run of $1$s. Then $1 \leq wt(\bfe') \leq 2$ since $e_1=0$. Suppose further that $\bfe'=(e_1', \ldots, e_{n-1}')$ has its first $1$ at position $i$ where $1 \leq i \leq n-1$. Since $i$ is the location of the first $1$ in $\bfe'$, then $e_{i} \neq e_{i+1}$ and so $e_{i}=0, e_{i+1}=1$ (since $e_1 = 0$). However, if $e_{i+1}=1$, then $x_{i} \neq x_{i+1}$ and so both $x_{i}'=e_{i}' = 1$. Since $wt(\bfe') \leq 2$, we have just shown that $|supp(\bfx') \cap supp(\bfe')| \geq 1$, and so the base case is complete.

We now assume that for any length-$n$ $\bfe_{\bfx}$, if $\bfe_{\bfx}$ has $k$ runs of $1$s, then $r(\bfy) \leq r(\bfx)$ where $1 \leq k \leq \lfloor \frac{n}{2} \rfloor$. Consider the case where $\bfe_{\bfx}$ has $k+1$ runs of $1$s. Suppose the $k$-th run of $1$s in $\bfe_{\bfx}$ has its final $1$ in position $j$ where $2 \leq j \leq n-2$. Thus, $e_{j+1}=0$. For shorthand denote $\bfe_1=(e_1, \ldots, e_{j+1})$, $\bfe_2=(e_{j+1}, \ldots, e_{n})$, $\bfx_1=(x_1, \ldots, x_{j+1})$, $\bfx_2=(x_{j+1}, \ldots, x_{n})$, $\bfe'_1 = \Psi(\bfe_1)$, $\bfe'_2 = \Psi(\bfe_2)$, $\bfx_1' = \Psi(\bfx_1)$, and $\bfx_2' = \Psi(\bfx_2)$. Notice that the vectors $\bfe'$ and $\bfx'$ can be written as the concatenation of two vectors where $\bfe' = ( \bfe_1', \bfe_2' )$ and $\bfx' = ( \bfx_1', \bfx_2' )$ where $\bfe_1$ is a $t$-grain-error for $\bfx_1$ with $k$ runs of $1$s and $\bfe_2$ is a $t$-grain-error for $\bfx_2$ with a single run of $1$s. By the inductive assumption, $|supp(\bfx_1') \cap supp(\bfe_1')| \geq \frac{wt(\bfe_1')}{2}$ and $|supp(\bfx_2') \cap supp(\bfe_2')| \geq \frac{wt(\bfe_2')}{2}$. Combining these two statements gives the desired result that $|supp(\bfx') \cap supp(\bfe')| \geq \frac{wt(\bfe')}{2}$ and so the proof is complete.

\end{proof}

The following lemma follows from the structure of grain-errors.

\begin{lemma}\label{lem:inf} For any two vectors $\bfx, \bfu \in GF(2)^n$, suppose that for some $1 \leq i \leq n-1$, 
\begin{enumerate}
\item $( x_{i}, x_{i+1} ) = (0, 0), ( u_{i}, u_{i+1} )=(1, 1)$ or 
\item $( x_{i}, x_{i+1} ) = (1, 1), ( u_{i}, u_{i+1} )=(0, 0)$.
\end{enumerate}
Then, $\cB_{t,G}(\bfx) \cap \cB_{t,G}(\bfu) = \emptyset$. \end{lemma}
\begin{proof}
Let $\bfy_1 = \bfx + \bfe_{\bfx}$ and $\bfy_2 = \bfu + \bfe_{\bfu}$. Since $\bfx$ and $\bfu$ differ at position $i+1$ then in order for $\bfy_1 = \bfy_2$, an error must occur at position $i+1$ in either $\bfx$ or $\bfu$ but not both. However, a grain-error can never change the information at position $i+1$ in either $\bfx$ or $\bfu$ since both $\bfx$ and $\bfu$ store the same information in positions $i$ and $i+1$ by the conditions in the statement of the lemma.
\end{proof}

We now prove the final lemma for this subsection.

\begin{lemma}\label{lem:even} Suppose $\cC$ is a $t$-grain-error-correcting code of length $n$ with the maximum possible cardinality. Then, $|\cC|$ is an even number.  \end{lemma}
\begin{proof} Let $\cG_{G}=(V,E)$ be a simple graph (i.e., a graph with undirected edges, with no parallel edges and no self loops), where $V = GF(2)^n$ and $E=\{ (\bfv_1, \bfv_2)\in V^2 : \cB_{t,G}(\bfv_2) \cap \cB_{t,G}(\bfv_1) \neq \emptyset \}$. A code $\cC$ is a $t$-grain-error-correcting code if and only if the set $\cC$ is an independent set in the graph $\cG$. Let $V_0 = \{ \bfv \in V : v_1=0 \}$ and $V_1 = \{ \bfv \in V : v_1=1 \}$. From Claim~\ref{cl:mingrain} there are no edges between any of the vertices in $V_0$ and $V_1$.  Consider the subgraphs $\cG_{0}=(V_0, E_0)$ where $E_0$ consists of all the edges in $E$ between the vertices in $V_0$, and $\cG_{1}=(V_1, E_1)$ where $E_1$ consists of all the edges in $E$ between vertices in $V_1$. Again from Claim~\ref{cl:mingrain}, $\cG_0$ and $\cG_1$ are isomorphic so that the maximum size of an independent set from $\cG_0$ is equal to the maximum size of an independent set from $\cG_1$ and thus the statement in the lemma holds.
%
%
%
%
\end{proof}

The next claim will be used later in Section~\ref{sec:constructions} for constructing grain codes.

\begin{claim}\label{cl:gmcodes} Suppose $\cC_{M}$ is a $t$-mineral-error-correcting code of length $n$. Let $\cC$ be the code that is the result of prepending an arbitrary bit to the beginning of every codeword in $\cC_{M}$. Then, $\cC$ is a length-$(n+1)$ $t$-grain-error-correcting code of size $2 |\cC_M|$. \end{claim}

%

\subsection{Tools for computing upper bounds}\label{subsec:toolsub}
In this subsection, we briefly review some of the tools used in Section~\ref{sec:bound} for computing a non-asymptotic upper bound on the cardinality of grain-error-correcting codes. 
We begin by revisiting some of the notation and results from \cite{Negar}.  

\begin{definition} A \textit{\bf{hypergraph}} $\cH$ is a pair $( \cX, \cE)$, where $\cX$ is a finite set and $\cE$ is a collection of nonempty subsets of $\cX$ such that $\cup_{E \in \cE} E = \cX$. The elements of $\cE$ are called \textit{\bf{hyperedges}}. \end{definition}

\begin{definition} A \textit{\bf{matching}} of a hypergraph $\cH=( \cX, \cE )$ is a collection of pairwise disjoint hyperedges $E_1, \ldots, E_j \in \cE$. The \textit{\bf{matching number}} of $\cH$, denoted $\nu(\cH)$, is the largest $j$ for which such a matching exists. \end{definition} 

As will be described shortly, the following can be interpreted as the dual of the matching of a hypergraph.

\begin{definition} A \textit{\bf{transversal}} of a hypergraph $\cH = (\cX, \cE)$ is a subset $T \subseteq \cX$ that intersects every hyperedge in $\cE$. The \textit{\bf{transversal number}} of $\cH$, denoted by $\tau(\cH)$, is the smallest size of a transversal.  \end{definition}

Let $\cH$ be a hypergraph with vertices $x_1, \ldots, x_n$ and hyperedges $E_1, \ldots, E_m$. The relationships contained within $\cH$ can be interpreted through a matrix $A \in \{0, 1\}^{n \times m}$, where

$$ A(i,j) = \begin{cases}
   1 & \text{if } x_i \in E_j, \\
   0       & \text{otherwise},
  \end{cases}$$
for $1 \leq i \leq n, 1 \leq j \leq m$. Cast in this light, the matching number and the transversal number can be derived using linear optimization techniques.

\begin{lemma}(cf. \cite{Negar}) The matching number and the transversal number are the solutions of the integer linear programs:
\begin{align}\label{optimize1}
\nu(\cH) &= \max \{ {\bf1^T} {\bfz} | A {\bfz} \leq {\bf1}, z_j \in \{0,1\}, 1\leq j \leq m \}, \text{ and}\\
\tau(\cH) &= \min \{ {\bf1^T} {\bfu} | A^T {\bfu} \geq {\bf1}, u_i \in \{ 0, 1\}, 1 \leq i \leq n \},
\end{align}
where $\bf1$ denotes a column vector of all $1$s of the appropriate dimension.
\end{lemma} 

Relaxing the condition that the solutions to the programming problem are comprised of $0$s and $1$s, we have the following problems:
\begin{align}
\nu^{*}(\cH) &= \max \{ {\bf1^T} {\bfz} | A {\bfz} \leq {\bf1}, {\bfz} \geq 0\}\label{optimize2}, \text{ and} \\
\tau^{*}(\cH) &= \min \{ {\bf1^T} {\bfu} | A^T {\bfu} \geq {\bf1}, {\bfu} \geq 0 \}\label{optimize3}.
\end{align}
Clearly $\nu(\cH) \leq  \nu^{*}(\cH)$ and $\tau(\cH) \geq  \tau^{*}(\cH)$. Since (\ref{optimize2}) and (\ref{optimize3}) are linear programs, they satisfy strong duality \cite{boyd} and $\nu^{*}(\cH)=\tau^{*}(\cH)$. Thus, combining these inequalities leads us to $ \nu(\cH) \leq \tau^{*}(\cH)$ \cite{Negar}.

\subsection{Graph notation}\label{sec:confusabilitygraph}
In this subsection, we describe graph notation from \cite{West} that will be used in Section~\ref{subsec:improvedgrains} and Section~\ref{sec:improvedsingle}. Let $\cG=(V,E)$ be a simple graph. 
A vertex $\bfv_1 \in V$ is adjacent to another vertex $\bfv_2 \in V$ if there exists an edge between them. The degree of a vertex is the number of its adjacent vertices and the maximum degree of a vertex in $\cG$ is denoted $\Delta(\cG)$.

A \textit{\textbf{$\bfk$-coloring}} is a mapping $\Phi : V \to \{ 0, 1, \ldots, k-1 \}$ of numbers to each vertex such that the same number is never assigned to adjacent vertices. The \textit{\textbf{chromatic number}} of a graph, denoted by $\chi(\cG)$, is the smallest $k$ for which a $k$-coloring exists. A \textit{\textbf{clique}} is a set of vertices in $\cG$ that are all adjacent. The size of the largest clique in a graph $\cG$ is denoted $\varsigma(\cG)$. It is known that for a graph $\cG$, $\chi(\cG)$ is such that $\varsigma(\cG) \leq \chi(\cG) \leq \Delta(\cG) + 1$ \cite{West}. Each collection of vertices that share the same number (under some fixed $k$-coloring) is referred to as a \textit{\textbf{color class}}.

%
%

\subsection{Distance metrics and group codes}\label{subsec:distances}
In this subsection, we introduce some distance metrics that are used in Section~\ref{sec:constructions} to construct grain-error-correcting codes. In addition, we define group codes that will serve as the foundation of the single grain-error-correcting codes introduced in Section~\ref{subsec:singlegrain}.

\begin{definition} Suppose $\bfx,\bfy \in GF(2)^n$. Their \textit{\bf{Hamming distance}} is denoted $d_H(\bfx,\bfy)=| \{ i : x_i \neq y_i \} |$. \end{definition}


\begin{definition} Suppose $\bfx=(x_1,\ldots, x_{n}),\bfy=(y_1,\ldots, y_{n}) \in GF(2)^n$. For $1 \leq i \leq n$, $N(\bfx,\bfy)=| \{i : x_i > y_i \}|$. \end{definition}

\begin{definition} (cf. \cite{cr}) Suppose $\bfx,\bfy$ are two vectors in $GF(2)^n$. Their \textit{\bf{asymmetric distance}} is denoted $d_A(\bfx,\bfy)= \max \{ N(\bfx,\bfy), N(\bfy,\bfx) \}$. \end{definition}

We say that a code $\cC$ has \textit{\textbf{minimum Hamming distance}} $d_{H}(\cC)$ if $d_{H}(\cC)$ is the smallest Hamming distance between any two distinct codewords in $\cC$. Similarly, we say that a code $\cC$ has \textit{\textbf{minimum asymmetric distance}} $d_A(\cC)$ if $d_A(\cC)$ is the smallest asymmetric distance between any two distinct codewords in $\cC$.

Suppose $\cA$ is an additive Abelian group of order $n+1$ and suppose $( \tilde{g}_1, \ldots, \tilde{g}_{n} )$ is a sequence consisting of the distinct non-zero elements of $\cA$. For every $a \in \cA$, we define a \textit{\textbf{group code}} $\tilde{\cC}^{\cA}_{a}$ to be
$$ \tilde{\cC}^{\cA}_a = \{ \bfx \in GF(2)^{n} | \sum_{k=1}^{n} x_k \tilde{g}_k = a \}. $$ 

Without loss of generality, we assume the Abelian groups we deal with in this paper are additive. The resulting code construction was shown in \cite{cr} to have $d_H(\tilde{\cC}^{\cA}_a) \geq d_A(\tilde{\cC}^{\cA}_a) \geq 2$. We include the following example for clarity.

\begin{example} Let $\cA$ be the additive Abelian group $\mathbb{Z}_3$ so that $( \tilde{g}_1, \tilde{g}_2 ) = ( 1, 2 )$. Then, the group $\cA$ partitions the space $GF(2)^2$ into $3$ group codes.
\begin{align*}
\tilde{\cC}^{\mathbb{Z}_{3}}_0 &= \{ (0, 0), (1, 1) \}, \\
\tilde{\cC}^{\mathbb{Z}_{3}}_1 &= \{ (1, 0) \}, \\
\tilde{\cC}^{\mathbb{Z}_{3}}_2 &= \{ (0, 1) \}.
\end{align*}
\end{example}

An \textit{\textbf{elementary Abelian group}} is a finite Abelian group where every non-identity element in the group has order $p$, where $p$ is a prime. For shorthand, the elementary Abelian group of size $p^r$ (for a prime $p$ and a positive integer $r$) is referred to as an \textit{\textbf{elementary Abelian $\bfp$-group}} \cite{rotman}.

\subsection{Discrete Fourier analysis}\label{subsec:fourier}
In this subsection, we briefly review some of the tools that will be used in Section~\ref{sec:cardinalities} to derive lower bounds on the cardinalities of code constructions. The notation adopted is similar to the notation used in \cite{Mceliece}.

Let $p$ be a prime number and suppose $\zeta_p$ denotes the complex primitive $p$-th root of unity and suppose $r$ is some positive integer. Let $\cA$ refer to the additive Abelian group $(\mathbb{Z}_p)^r = \underbrace{\mathbb{Z}_p \times \cdots \times \mathbb{Z}_p}_{\text{$r$ times}}$. The operator $ {\langle \bfg, \bfh \rangle}$ takes two elements $\bfg=(g_1, \ldots, g_r), \bfh=(h_1, \ldots, h_r) \in \cA$ and maps them into a complex number as follows  
$$  {\langle \bfg, \bfh \rangle} = \prod_{i=1}^r (\zeta_p)^{g_i h_i} = (\zeta_p)^{\sum_{i=1}^r g_i h_i} = (\zeta_p)^{\bfg^T \cdot \bfh}. $$
Let $f({\bfg})$ be any function that maps elements of $\cA$ into the complex plane. The \textbf{\textit{Fourier transform}} $\hat{f}$ of $f$ is defined as 
$$ \hat{f}(\bfh) = \sum_{{\bfg} \in \cA}  {\langle \bfh, -\bfg \rangle} f(\bfg) $$
and the \textit{\textbf{inverse Fourier transform}} is defined as 
$$ f(\bfg) = \frac{1}{p^{r}} \sum_{{\bfh} \in \cA}  {\langle \bfh, \bfg \rangle} \hat{f}(\bfh).$$

\section{Upper Bounds on Grain-Error Codes}\label{sec:bound}

In this section, we use linear programming methods to produce a closed-form upper bound on the cardinality of a $t$-grain-error-correcting code. The approach is analogous to that found in \cite{Negar} where upper bounds were computed for the deletion channel and in \cite{Kashyap} where upper bounds were derived for the non-overlapping grain-error model. Recall, our objective is to compute upper bounds for the overlapping grain-error model. 

The approach is the following. First, the vector space from which codewords are chosen, is projected onto a hypergraph. Then, an approximate solution to a matching problem is derived. 
Recall from Section~\ref{subsec:errors} that the maximum size of a $t$-grain-error-correcting-code of length $n$ will be referred to as $M(n,t)$. 

Let $\cH_{t,n}$ denote the hypergraph for a $t$-grain-error-correcting code. More formally, let
$$ \cH_{t,n} = ( GF(2)^n, \{ \cB_{t,G}(\bfx) | \bfx \in GF(2)^n \} ). $$
In this graph, the vertices represent candidate codewords and the hyperedges represent vectors that result when $t$ or fewer grain-errors occur in any of the candidate codewords.

As in \cite{Negar}, $\nu^{*}(\cH_{t,n})$ is an upper bound on $M(n,t)$ and will be derived by considering the dual problem defined in $(4)$. The problem is to find a function $w : GF(2)^n \to \mathbb{R^{+}}$ such that 
\begin{equation}\label{eq:trans}
\tau^{*}(\cH_{t,n}) = \hspace{-2ex}\min_{w : GF(2)^n \to \mathbb{R^{+}}} \hspace{-1ex}\Bigg\{\hspace{-0.5ex} \sum_{\bfy \in GF(2)^n} \hspace{-3ex}w(\bfy) \ | \ \hspace{-3ex} \sum_{\bfy \in \cB_{t,G}(\bfx) } \hspace{-3ex} w(\bfy) \geq 1, \forall \bfx \in GF(2)^{n} \Bigg\}.
\end{equation}


%
%

We are now ready to state the main result of the section.

\begin{theorem}\label{th:maxcard} For positive integers $n,t$ where $t<n$, $$ M(n,t) \leq 2 \sum_{k=0}^{n-1} \nchoosek{n-1}{k} \frac{1}{\sum_{j=0}^{\min\{t,k\}} \nchoosek{k}{j}}. $$ \end{theorem}
\begin{proof} In order to prove the result, we must assign values for $w(\bfy)$ such that the constraint in (\ref{eq:trans}) is satisfied. Let $w(\bfy) = \frac{1}{ b_{t,n}(\bfy) }$ where $b_{t,n}(\bfy)$ is computed as in Lemma~\ref{lem:runs}.
Note that 
\begin{align*}
&\sum_{\bfy \in \cB_{{t,G}}(\bfx) } w(\bfy) =  \sum_{\bfy \in \cB_{{t,G}}(\bfx) } \frac{1}{ b_{t,n}(\bfy)}.
\end{align*}
From Lemma~\ref{lem:runsdec}, for any $\bfy \in \cB_{{t,G}}(\bfx)$, $b_{t,n}(\bfy) \leq b_{t,n}(\bfx)$, so we have
\begin{align*}
&\sum_{\bfy \in \cB_{{t,G}}(\bfx) } \frac{1}{ b_{t,n}(\bfy)} \geq  \sum_{\bfy \in \cB_{{t,G}}(\bfx) } \frac{1}{ b_{t,n}(\bfx)} = b_{t,n} (\bfx) \frac{ 1 }{b_{t,n}(\bfx)} = 1.
\end{align*}

The theorem statement now follows from the bound on $\sum_{\bfy \in GF(2)^n} w(\bfy)$: Since the number of length-$n$ vectors with $k$ runs is $2 \nchoosek{n-1}{k-1}$ and $b_{t,n}=\sum_{j=0}^{\min\{t,k-1\}} \nchoosek{k-1}{j}$ from Lemma~1, we have $$M(n,t) \leq 2 \sum_{k=1}^{n} \nchoosek{n-1}{k-1} \frac{1}{\sum_{j=0}^{\min\{t,k-1\}} \nchoosek{k-1}{j}},$$ which, after reindexing the parameter $k$, is the statement in the theorem.
\end{proof}

Theorem~\ref{th:maxcard} gives an explicit upper bound on $M(n,t)$ for all $n$ and $t$. However, providing an explicit expression (without summations) is still not easy to derive. In the following, we present non-asymptotic bounds for $t=1,2,3$. The bounds will then be compared against the existing bounds in \cite{artyom} for $t=1,2,3$. Note that the overlapping and non-overlapping grain-error models coincide for the case where $t=1$. The following corollary was also derived in \cite{Kashyap} in the context of the non-overlapping grain-error model. It is the result of combining Theorem~\ref{th:maxcard} for the case where $t=1$ with Lemma~\ref{lem:even}. Recall, $M(n,t)$ refers to the maximum size of a $t$-grain-error-correcting code.

\begin{corollary}\label{cor:card1} For $n \geq 1$, $M(n,1) \leq 2 \lfloor \frac{2^{n+1} - 2}{2n} \rfloor$. \end{corollary}


For the case of $t=2$, we make use of the following claims which can be proven using induction. The details are included in Appendix~\ref{app:m3}.

\begin{claim}\label{cl:2} For $n \geq 2$, $$\sum_{k=2}^n \frac{1}{k+1} \nchoosek{n}{k} = \frac{1}{n+1} \left(2^{n+1} - 2 - \frac{3n}{2} - \frac{n^2}{2} \right).$$ \end{claim}
\begin{claim}\label{cl:3} For $n \geq 17$,
$$\sum_{k=1}^n \frac{1}{k} \nchoosek{n}{k} \leq \frac{2^{n+1}}{n-1-\frac{2}{n-5}+\frac{1}{n^2}}.$$ \end{claim}

We now derive the bound for $M(n,2)$, the maximum size of a $2$-grain-error-correcting code, which is non-asymptotic and explicit.

\begin{lemma}\label{lem:m2} For $n \geq 18$, $M(n,2) \leq 2  \Big\lfloor \frac{2^{n+2}(2+\frac{2}{n-6})}{2n(n-3)}  \Big\rfloor$. \end{lemma}

\begin{proof}
From Theorem~\ref{th:maxcard} we have
\begin{small}
\begin{align*}
& M(n,2) \leq 2 \sum_{k=0}^{n-1} \nchoosek{n-1}{k} \frac{1}{\sum_{j=0}^{\min\{2,k\}} \nchoosek{k}{j}}  \\
& = 2 + n-1+2\sum_{k=2}^{n-1}  \nchoosek{n-1}{k} \frac{1}{1+k+\binom{k}{2}} \\
& \leq n+1 + 4\sum_{k=2}^{n-1}  \nchoosek{n-1}{k} \left( \frac{1}{k}-\frac{1}{k+1} \right) \\
& = n+1 + 4\sum_{k=2}^{n-1}  \nchoosek{n-1}{k} \frac{1}{k} -4\sum_{k=2}^{n-1}  \nchoosek{n-1}{k} \frac{1}{k+1}.
\end{align*}
\end{small}
From Claims~\ref{cl:2} and~\ref{cl:3} we have 
\begin{small}
\begin{align*}
M(n,2) &\leq n + 1 + 4 \left( \frac{2^n}{n-2-\frac{2}{n-6} + \frac{1}{(n-1)^2}}  -n+1 \right) \\
& -\frac{4}{n}\left( 2^n-2-\frac{3(n-1)}{2}-\frac{(n-1)^2}{2} \right)\\
& =\frac{2^{n+2}}{n-2-\frac{2}{n-6}}-\frac{2^{n+2}}{n} -n + 7 +\frac{4}{n}\\
& \leq \frac{2^{n+2}(2+\frac{2}{n-6})}{n(n-3)}.
\end{align*}
\end{small}
From Lemma~\ref{lem:even} $M(n,2)$ must be an even integer and so $M(n,2) \leq 2  \Big\lfloor \frac{2^{n+2}(2+\frac{2}{n-6})}{2n(n-3)}  \Big\rfloor $.
\end{proof}

For $t=3$, the upper bound is stated as a lemma. The details can be found in Appendix~\ref{app:m3}.

\begin{lemma}\label{lem:ub3}For $n \geq 24$, $$M(n,3) \leq 2 \Big\lfloor 2^n\left( \frac{\frac{36n}{n-7} + 18 - \frac{3n(n-1)}{(n-2)^2}  + \frac{12}{n-7}}{2n(n-1)(n-3-\frac{2}{n-7} + \frac{1}{(n-2)^2}} \right) \Big\rfloor.$$ \end{lemma}

We briefly note that Lemma~\ref{lem:ub3} provides a looser asymptotic upper bound on $M(n,t)$ for the case where $t=3$ than the bound provided in \cite{arya}. However, Lemma~\ref{lem:ub3} has the advantage of providing a non-asymptotic expression. 
We briefly note that asymptotically Lemma~\ref{lem:ub3} provides a looser upper bound on $M(n,t)$ for the case where $t=3$ than the bound provided in \cite{arya}. However, Lemma~\ref{lem:ub3} has the advantage of providing an explicit expression that holds for a broad range of $n$.

 {In Table~\ref{table:resultsub}, we illustrate the result for $M(n,t)$ for small $n$ and $t=1,2,3$ using Theorem~\ref{th:maxcard}. Each sub-column in Table~\ref{table:resultsub} consists of a pair of entries, the left entry corresponding to the previous best upper bound (and labeled 'Prev UB'), and the right entry corresponding to the upper bound offered by Theorem 1 (and labeled 'New UB'). Entries under 'Prev UB' are taken from \cite{artyom} as follows. For $M(n,t)$ where $n \leq 8$, the entries are taken from Table~II in \cite{artyom}; the entry for  $M(9,3)$ was also taken from the same table. The remaining entries for the previous upper bounds on $M(n,t)$ are taken from Theorem 3.1 in \cite{artyom}.}

 {For all values of $10 \leq n \leq 20$ the bound in Theorem~\ref{th:maxcard} was tighter (as can be seen in Table~\ref{table:resultsub}) than the corresponding bound in \cite{artyom}. In addition, the bounds in this section have the advantage of being explicit. We remark that in a more general case, it is difficult to compare our bounds to those in \cite{artyom} since the bounds in \cite{artyom} require finding a parameter $\rho$ where $\rho$ is the largest integer satisfying  $\sum_{k=0}^{\rho-1} \nchoosek{n-1}{k} \sum_{j=0}^{\min(t,k)} \nchoosek{k}{j} \leq 2^{n-1}$.}


\begin{table}
\center{\caption{Comparison of Existing Upper Bounds from \cite{artyom} and New Upper Bounds from Theorem~\ref{th:maxcard}}\label{table:resultsub}
\begin{tabular}{|c|cc|cc|cc|}
	\hline
Length  & \multicolumn{2}{c|}{$M(n,1)$}     &  \multicolumn{2}{c|}{$M(n,2)$}                                                          			&  \multicolumn{2}{c|}{$M(n,3)$}   \\
	     & Prev  {UB} & New  {UB} 	& Prev  {UB} & New  {UB}	& Prev  {UB} & New  {UB} \\
	\hline
3           &   4  &  {4}                    & - & -                                                                     & - & -            \\
4           &   6 &  {6}                    & 4 & 6                                         & - & -            \\
5           &   8 & 12                                              & 8 &  10                                       & - & -  \\
6           &   16 & 20                                            & 10 & 14                                     & 8 &14  \\
7           &   26 & 36                                            & 16 & 24                                     & 16 & 22  \\
8           &   44 & 62                                            & 22 & 38                                     & 18 & 34  \\
9           &   150 &  {112}           & 110 &  {62}        & 32 & 52 \\
10         &   278 &  {204}           & 190 &  {102}      & 178 &  {80}  \\
11         &   506 &  {372}            & 346 &  {168}     & 316 &  {126}  \\
12         &   942 &  {682}            & 582 &  {280}     & 528 &  {198}  \\
13         &   1760 &  {1260}        & 1010 &  {476}   & 870 &  {312}  \\
14         &   3256 &  {2340}        & 1818 &  {814}   & 1498 &  {496}  \\
15         &   6148 &   {4368}        & 3246 &  {1406}  & 2682 &  {800}  \\
16         &   11532 &  {8190}      & 5646 &  {2448}   & 4512 &  {1300}  \\
17         &   21654 &  {15420}    & 10168 & {4302} & 7590 &  {2132}  \\
18         &   41340 &   {29126}    & 18852 & {7612} & 13300 &  {3528}  \\
19         &   77792 &  {55188}    & 32962 & {13560}& 24178 &  {5892}  \\
20         &   147788 &  {104856}& 59518 &  {24306}& 40724 &  {9920}  \\
	\hline
\end{tabular}\label{table:Mnt}}
\end{table}

\section{Grain-Error Code Constructions}\label{sec:constructions}

In the previous section, the focus was on upper bounds for grain-error-correcting codes. In this section, we turn to code constructions. We will compare the codes proposed in this section to the upper bounds derived in the previous section.

This section is divided into three subsections. In Section~\ref{subsec:singlegrain}, we consider a group-theoretic construction for single-grain codes. In Section~\ref{subsec:improvedgrains}, we generalize the construction from~\ref{subsec:singlegrain}. Using this generalization, Section~\ref{subsec:improvedgrains} then identifies better codes that correct single grain-errors for certain code lengths. Section~\ref{subsec:manygrains} considers constructions for codes that can correct multiple grain-errors.

\subsection{Single-grain codes}\label{subsec:singlegrain}

We begin by proving some sufficient conditions for a code to correct a single grain-error. Then, we provide a group-theoretic code construction that satisfies these conditions. 
The codes presented in this section provide the largest known cardinalities for all code lengths greater than $16$.

Combining Lemma~\ref{lem:inf} with Definition~\ref{def:grainerror}, the following claim can be verified. Recall that $d_H$ and $d_A$ refer to the Hamming distance and the asymmetric distance, respectively.

\begin{claim}\label{cl:distance} A code $\cC$ is a single-grain code if for every pair of distinct codewords $\bfx,\bfy \in \cC$ if at least one of the following holds:
\begin{enumerate}
\item $d_H(\bfx,\bfy)=1$ and $x_1 \neq y_1$.
\item $d_H(\bfx, \bfy)=2$ and for some $1 < i \leq n-1$, 
\begin{enumerate}
\item $( x_{i}, x_{i+1} ) = (0, 0), ( y_{i}, y_{i+1} )=(1, 1)$ or 
\item $( x_{i}, x_{i+1} ) = (1, 1), ( y_{i}, y_{i+1} )=(0, 0)$.
\end{enumerate}

\item $d_H(\bfx,\bfy) \geq 3$.
\end{enumerate}
\end{claim}

We are now ready to state our code construction. For any additive Abelian group referred to in the subsequent discussion, the identity element will be denoted as $0$ and will be referred to as the zero element.

\begin{Construction}\label{grain-construct} Let $\cA$ represent an additive Abelian group of size $n$. Suppose the sequence $\cS=( g_1, g_2, \ldots, g_{n} )$, which contains every element of $\cA$ once, is ordered as follows: 
\begin{enumerate}
\item $g_1=0$,
\item for any $1 < i \leq n$, the elements $g_i$ and $-g_i$ (if $-g_i \neq g_i$) are adjacent.
\end{enumerate}
For any element $a \in \cA$, let
\begin{align}\label{eq:ggraincode}
\cC^{\cA}_a = \{ \bfx \in \{0,1\}^n : \sum_{k=1}^{n} x_k g_k = a \}. 
\end{align}
\end{Construction}

The following example illustrates Construction~\ref{grain-construct}.

\begin{example} Let $\cA$ denote the additive Abelian group $\mathbb{Z}_3$. Suppose Construction~\ref{grain-construct} is used to create a code where $\cS = ( g_1, g_2, g_3 ) = ( 0, 1, 2)$. Then, the group $\cA$ partitions the space $GF(2)^3$ into $3$ single-grain codes.
\begin{align*}
\cC^{\mathbb{Z}_3}_{0} &= \{ (0, 0, 0),(1, 0, 0), (0, 1, 1), (1, 1, 1) \}, \\
\cC^{\mathbb{Z}_3}_{1} &= \{ (0, 1, 0), (1, 1, 0) \}, \\
\cC^{\mathbb{Z}_3}_{2} &= \{ (0, 0, 1), (1, 0, 1) \}.
\end{align*}
\end{example}

The correctness of Construction~\ref{grain-construct} is proven next.

\begin{theorem}\label{th:sec} A code $\cC^{\cA}_a$ created with Construction~\ref{grain-construct} is a single-grain code. \end{theorem}
\begin{proof} We will show that $\cC^{\cA}_a$ is a single-grain code by demonstrating that the conditions listed in Claim~\ref{cl:distance} hold for any pair of distinct codewords $\bfx,\bfy \in \cC^{\cA}_a$. Let $\tilde{\cC}^{\cA}_{a}$ be the group code created by using the same group and element $a$ as in $\cC^{\cA}_a$ so that $\tilde{\cC}^{\cA}_a$ has length $n-1$, and $\tilde{\cC}^{\cA}_a$ is obtained by shortening the codewords of $\cC^{\cA}_a$ on the first bit (i.e., by removing $x_1$, which multiplies $g_1=0$). Recall from Section~\ref{subsec:distances} that since $\tilde{\cC}^{\cA}_a$ is a group code, $d_{H}(\tilde{\cC}^{\cA}_a) \geq d_A(\tilde{\cC}^{\cA}_a) \geq 2$. 

Suppose $d_H(\bfx,\bfy) = 1$. Then, since $d_{H}(\tilde{\cC}^{\cA}_a) \geq 2$, it follows that if $d_H(\bfx,\bfy) = 1$, then $\bfx$ and $\bfy$ differ only in the first bit and so condition 1) from Claim~\ref{cl:distance} holds.

Suppose $d_H(\bfx,\bfy) = 2$. Since $d_{H}(\tilde{\cC}^{\cA}_a) \geq 2$, $\bfx$ and $\bfy$ do not differ in the first position, and there are two distinct indices $i,j$ ($2 \leq i,j \leq n$) where $x_i \neq y_i$ and $x_j \neq y_j$. Suppose, without loss of generality, that $N(\bfx,\bfy)=2$ and so $x_i=x_j=1$. Therefore, $g_i + g_j=0$, or $g_j=g_i^{-1}$. However, by condition 2) in Construction~\ref{grain-construct}, we have $|j-i|=1$ and so condition 2) from Claim~\ref{cl:distance} holds.

If $d_H(\bfx,\bfy)$ is not equal to $1$ or $2$ then $d_H(\bfx,\bfy) \geq 3$ and so condition 3) of Claim~\ref{cl:distance} holds.\end{proof}

The following corollary follows from the proof of Theorem~\ref{th:sec} and Claim~\ref{cl:mingrain}.

\begin{corollary}\label{cor:singlemin} Let $\cC^{\cA}_a$ be a single-grain code created according to Construction~\ref{grain-construct}. Let $\tilde{\cC}^{\cA}_a$ be the group code that is the result of shortening the codewords in $\cC^{\cA}_a$ on the first bit. Then $\tilde{\cC}^{\cA}_a$ is a single-mineral code. \end{corollary}


The following corollary provides upper and lower bounds on $|\cC^{\cA}_a|$.

\begin{corollary}\label{cor:singlemincard} Suppose $\cA$ is an Abelian group of size $n$ and $a \in \cA$. Then, for a code $\cC^{\cA}_a$ created according to Construction~\ref{grain-construct}, $|\cC^{\cA}_a| \leq |\cC^{\cA}_0|$. Furthermore, 
$$ \frac{2^n}{n} \leq | \cC^{\cA}_0| \leq \frac{2^n}{n} + \frac{(n-1) \cdot 2^{n/3}}{n}. $$
Equality holds on the left if and only if $|\cA|$ is a power of two. Equality holds on the right if and only if $\cA$ is an elementary Abelian $3$-group.
 \end{corollary}
\begin{proof} Since Construction~\ref{grain-construct} concatenates an arbitrary bit with a group code, it follows that if the underlying group code of length $n'=n-1$ has cardinality $| \tilde{\cC}^{\cA}_a |$, then the code $\cC^{\cA}_a$ created using the previous construction has $2|\tilde{\cC}^{\cA}_a|$ codewords. Then, since $|\tilde{\cC}^{\cA}_a| \leq |\tilde{\cC}^{\cA}_0|$  (\cite{cr}, Theorem 9), $|\cC^{\cA}_a| \leq |\cC^{\cA}_0|$. Furthermore, from (\cite{Mceliece}, Corollary 2) $| \tilde{\cC}^{\cA}_0 | \leq \frac{1}{n'+1} \left( 2^{n'} + n' 2^{(n'-2)/3} \right)$ with equality if and only if $\cA$ is an elementary Abelian $3$-group. Replacing $n=n'+1$ and multiplying $| \tilde{\cC}^{\cA}_0 |$ by $2$ then gives the upper bound stated in the corollary. From (\cite{Mceliece}, Corollary 1), $| \tilde{\cC}^{\cA}_0 | = \frac{2^{n'}}{n'+1}$ if and only if $n'+1$ is a power of $2$. Replacing $n=n'+1$ and multiplying $| \tilde{\cC}^{\cA}_0 |$ by $2$ gives that $| \cC^{\cA}_0| = \frac{2^n}{n}$ when $n$ is a power of two. Since Construction~\ref{grain-construct} partitions the space $GF(2)^n$ into $n$ binary single-grain codes, it follows that for any $n$ we have $| \cC^{\cA}_0| \geq \frac{2^n}{n}$ and so the statement in the corollary follows.
\end{proof}

In \cite{artyom}, a single-grain code construction was given that produced codes of length $n=2^m-1$ with $\frac{2^{n}}{n+1} + 2^{\frac{(n-1)}{2}}$ codewords where $m$ is a positive integer. In \cite{arya}, a single-grain code construction was proposed that resulted in codes of length $n$ where $n=2^r$ (where $r$ is a positive integer) that contained $\frac{2^n}{n}$ codewords.

Our construction extends for any $n$ (via the set $\cA = \mathbb{Z}_n$). When $n$ is a power of $2$, Construction~\ref{grain-construct} produces codes with the same cardinality as the codes presented in \cite{arya}. Furthermore, for codes of length $n$ where $n$ is not a power of $2$, Construction~\ref{grain-construct} provides codebooks with cardinalities strictly greater than $\frac{2^n}{n}$ by Corollary~\ref{cor:singlemincard}.

Since, for large $n$, 
$$ \frac{2^n}{n} > \frac{2^{n}}{n+1} + 2^{\frac{n-1}{2}}, $$
Construction~\ref{grain-construct} improves upon the state of the art when $n$ is not a power of $2$ and $n \geq 15$.

In the next subsection, we provide a generalization of Construction~\ref{grain-construct}. We then derive constructions for single-grain codes that have even larger cardinalities and consider codes capable of correcting more than a single grain-error.

\subsection{Improved grain codes using mappings}\label{subsec:improvedgrains}

In \cite{grassl}, the authors make the observation that a single-asymmetric error-correcting code (and in particular a group code) can be constructed by defining a code over pairs of binary elements. Consider the map $\Gamma : \{ 0,1 \}^2 \to GF(3)$, which is defined as follows: 
\begin{align}\label{eq:map}
(0,0) \to 0, (0,1) \to 1, (1,0) \to 2, (1,1) \to 0.  
\end{align}
Note that the map is not one-to-one since both $(0,0)$ and $(1,1)$ map to $0$. With a slight abuse of notation, if the map $\Gamma$ is applied to a binary vector of even length then it is simply applied to each pair of consecutive elements at a time (e.g., $ \Gamma(0,1,0,0) =  (\Gamma(0,1),  \Gamma(0,0))$). Furthermore, (again with a slight abuse of notation) if the $\Gamma$ map is applied to a set of vectors, it returns a set of ternary vectors that are the result of applying the map to each vector in the set. Using this map, codes that correct asymmetric errors were proposed in \cite{grassl}. In the following, we illustrate how to generalize the ideas from \cite{grassl} (by using different mappings) to correct grain-errors. 

Let $\cG_{t,m}=(V,E)$ denote a simple graph (see Section~\ref{sec:confusabilitygraph}) where $V=GF(2)^m$. That is, the vertices of $\cG_{t,m}$ are the the vectors from $GF(2)^m$. For any $\bfx, \bfy \in V$, $(\bfx, \bfy) \in E$ if $\cB_{t,M}(\bfx) \cap \cB_{t,M}(\bfy) \neq \emptyset$. Recall from Section~\ref{sec:confusabilitygraph}, a mapping $\Phi_{t,m} : GF(2)^m \to \{ 0, 1, \ldots, p-1 \}$ is a $p$-coloring if it assigns different numbers to adjacent vertices. If the input to $\Phi_{t,m}$ is a vector of length $mn$, then the map is applied to each collection of $m$ consecutive bits at a time. For example, if $m=3$, then $\Phi_{t,3}(0, 0, 0, 1, 0, 1) = (\Phi_{t,3}(0, 0, 0)\ \Phi_{t,3}(1, 0, 1))$. 

%
%

\begin{Construction}\label{construct:color} Suppose $q \geq 2$ is a positive integer and $\Phi_{t,m} : GF(2)^m \to \{0, 1, \ldots, q-1\}$ is a $q$-coloring on $\cG_{t,m}$. Let $\cC_t$ be a $t$-unrestricted-error-correcting code over an alphabet of size $q$ of length $n$. Let 
\begin{align}\label{eq:colorconstruct}
\cC=\{ \bfx \in GF(2)^{mn} : \Phi_{t,m}(\bfx)  \in \cC_t \} .
\end{align}
\end{Construction}

\begin{remark} If $\cC$ is a code created according to Construction~\ref{construct:color}, then the map $\Phi_{t,m}$ can be interpreted as mapping the color classes of a {$q$}-coloring onto the symbols of a non-binary code $\cC_t$. This interpretation will be useful in Section~\ref{sec:improvedsingle}. \end{remark}

\begin{remark} As noted in \cite{grassl}, since a code created according to Construction~\ref{grain-construct} is a permutation of a group code, Construction~\ref{grain-construct} and Construction~\ref{construct:color} coincide for the case where $p=3$ and $\cC_1$ (from (\ref{eq:colorconstruct}) in Construction~\ref{construct:color}) is a single unrestricted-error-correcting code over $GF(3)$. \end{remark}

We now provide an example of a code created with Construction~\ref{construct:color}.

\begin{example}\label{ex:coloring} Let the map $\Gamma$ be as defined in (\ref{eq:map}). Note, from Lemma~\ref{lem:inf}, that the map $\Gamma$  is actually a coloring on $\cG_{t,2}$ where the set of vectors $GF(2)^2$ are partitioned into color classes as follows:
\begin{enumerate}
\item $\{ (0\ 0), (1\ 1) \}$,
\item $\{ (1\ 0) \}$,
\item $\{ (0\ 1) \}$.
\end{enumerate}
Let $\cC_t$ be a $t$-unrestricted-error-correcting code over $GF(3)$ of length $n$. Then the set of vectors
\begin{align}\label{eq:cond}
\cC= \{ \bfx \in GF(2)^{2n} : \Gamma( \bfx ) \in \cC_t \}
\end{align}
is a code created according to Construction~\ref{construct:color}.
\end{example}

\begin{remark}\label{rem:grainasym} We note that when $\cC_t$ is a single-unrestricted-error-correcting code, a code constructed according to Example~\ref{ex:coloring} coincides with the ternary construction from \cite{grassl} proposed in the context of asymmetric errors. \end{remark} 

%

We now prove that any code created according to Construction~\ref{construct:color} is a $t$-mineral-error-correcting code.

\begin{theorem}\label{th:mainconstruct} Let $\cC_t$ be a $t$-unrestricted-error-correcting code. Suppose $\cC$ is a code created according to Construction~\ref{construct:color} with $\cC_t$ as the constituent code. Then, $\cC$ is a $t$-mineral-error-correcting code. \end{theorem}
\begin{proof} The result will be proven by showing that for any codewords $\bfx, \bfy \in \cC$ where $\bfx \neq \bfy$, $\cB_{t,M}(\bfx) \cap \cB_{t,M}(\bfy) = \emptyset$. Consider two codewords $\bfx, \bfy \in \cC$ such that $\bfx \neq \bfy$. There are two cases to consider: either 1) $\Phi_{t,m}(\bfx) = \Phi_{t,m}(\bfy)$ or 2) $\Phi_{t,m}(\bfx) \neq \Phi_{t,m}(\bfy)$. Recall that, by construction, $\Phi_{t,m}(\bfx), \Phi_{t,m}(\bfy) \in \cC_t$.

Suppose $\Phi_{t,m}(\bfx) = \Phi_{t,m}(\bfy)$. Then, since $\bfx \neq \bfy$, there exists an index $i$ where $1 \leq i \leq n$ such that $\Phi_{t,m}(x_{(i-1)m + 1}, \ldots, x_{im}) = \Phi_{t,m}(y_{(i-1)m + 1}, \ldots, y_{im})$ but $(x_{(i-1)m + 1}, \ldots, x_{im}) \neq (y_{(i-1)m + 1}, \ldots, y_{im})$. For shorthand, let $\bfv_1=(x_{(i-1)m+1}, \ldots, x_{im})$ and $\bfv_2=(y_{(i-1)m+1}, \ldots, y_{im})$. Since $\Phi_{t,m}(\bfv_1) = \Phi_{t,m}(\bfv_2)$, the vectors $\bfv_1, \bfv_2$ map to the same color class under $\Phi_{t,m}$, which implies that $\bfv_1$ and $\bfv_2$ are not adjacent in $\cG_{t,m}$. By definition, if $\bfv_1, \bfv_2$ are not adjacent in $\cG_{t,m}$, $\cB_{t,M}(\bfv_1) \cap \cB_{t,M}(\bfv_2) = \emptyset$. Thus, for any $t$-mineral-errors (of length $m$) $\bfe_{\bfv_1}, \bfe_{\bfv_2}$, we have $\bfv_1 + \bfe_{\bfv_1} \neq \bfv_2 + \bfe_{\bfv_2}$. Then, there do not exist any $t$-mineral-errors $\bfe_{\bfx}, \bfe_{\bfy}$ such that $\bfx + \bfe_{\bfx} = \bfy + \bfe_{\bfy}$. Thus, $\cB_{t,M}(\bfx) \cap \cB_{t,M}(\bfy) = \emptyset$. 

Suppose now that $\Phi_{t,m}(\bfx) \neq \Phi_{t,m}(\bfy)$. Then, since $\Phi_{t,m}(\bfx), \Phi_{t,m}(\bfy) \in \cC_t$, there exists a set of at least $2t+1$ indices from $\{ 1, 2, \ldots, n \}$, denoted as $\cI$, such that $\forall j \in \cI$, $\Phi_{t,m}(x_{(j-1)m + 1}, \ldots, x_{jm}) \neq \Phi_{t,m}(y_{(j-1)m + 1}, \ldots, y_{jm})$. Since $\Phi_{t,m}(x_{(j-1)m + 1}, \ldots, x_{jm}) \neq \Phi_{t,m}(y_{(j-1)m + 1}, \ldots, y_{jm})$, $d_H((x_{(j-1)m + 1}, \ldots, x_{jm}), (y_{(j-1)m + 1}, \ldots, y_{jm})) \geq 1$ for every $j \in \cI$ and so $d_H(\bfx, \bfy) \geq 2t+1$. Thus, $\cB_{t,U}(\bfx) \cap \cB_{t,U}(\bfy) = \emptyset$ where $\cB_{t,U}$ denotes the error-ball for $t$ unrestricted-errors (as discussed in Section~\ref{subsec:errors}). From Claim~\ref{cl:ordering}, then $\cB_{t,M}(\bfx) \cap \cB_{t,M}(\bfy) = \emptyset$ as well and the proof is complete.
\end{proof}

Notice that according to Theorem~\ref{th:mainconstruct}, the code from Example~\ref{ex:coloring} is a $t$-mineral-error-correcting code. Corollary~\ref{cor:maingrains} follows from Claim~\ref{cl:gmcodes}.

Notice that the proof of Theorem~\ref{th:mainconstruct} relied on two properties of the error-ball $\cB_{t,M}$. In particular, the proof required that for any $\bfx=(\bfx_1, \ldots, \bfx_n) \in GF(2)^{mn}, \bfy=(\bfy_1, \ldots, \bfy_n) \in GF(2)^{mn}$:
\begin{enumerate}
\item $\cB_{t,M}(\bfx) \subseteq \cB_{t,U}(\bfx)$ and
\item If $\bfy \in \cB_{t,M}(\bfx)$, then for $1 \leq i \leq n$, $\cB_{t,M}(\bfy_i) \in \cB_{t,M}(\bfx_i)$.
\end{enumerate}
We note that many other channels satisfy the above two properties, such as the $Z$-channel. Construction~\ref{construct:color} can thus be used to generate codes that correct additional types of errors.

The next corollary follows from Theorem~\ref{th:mainconstruct}.

\begin{corollary}\label{cor:maingrains} Let $\cC'$ be a $t$-mineral-error-correcting code of length $mn$ created according to Construction~\ref{construct:color}. Then,
$$\cC=\{ \bfx \in GF(2)^{mn+1} : (x_2, \ldots, x_{mn+1}) \in \cC' \}$$
is a $t$-grain-error-correcting code. \end{corollary}

Although Construction~\ref{construct:color} provides a method to construct $t$-mineral-error-correcting codes, it is not straightforward to compute the sizes of the resulting codes because the color classes of the map $\Phi_{t,m}$ are not always of the same size. As a starting point, in this subsection we only considered single-mineral codes created using Construction~\ref{construct:color} with the map $\Gamma$ as described in Example~\ref{ex:coloring}. Even with the simple map $\Gamma$, computing the cardinalities of the resulting codes from Construction~\ref{construct:color} is not straightforward. In the following subsection, we analyze the codes from Example~\ref{ex:coloring} for arbitrary $t$.

Recall that from Remark~\ref{rem:grainasym}, the single asymmetric error-correcting codes proposed in \cite{grassl} (using the ternary construction) are a special case of Construction~\ref{construct:color}. Therefore, the codes from (Table II, column 4, \cite{grassl}) are single-mineral codes. Thus, we can obtain new single-grain codes by appending an information bit to these codes. The cardinalities displayed in the column titled `Current Lower Bound' (second column) of Table~\ref{table:results} for $9 \leq n \leq 15$ are the result of this operation. Note that the codes enumerated from \cite{grassl} were the result of a computerized search and to limit the search space, the search was only carried out on codes of length at most 15. For $n \geq 16$, the cardinalities in the second column of Table~\ref{table:results} can be obtained from Construction~\ref{grain-construct} using the group codes found in Table~1 in \cite{cr}. The first column in Table~\ref{table:results} (labeled 'Previous Lower Bound') shows the cardinalities of the largest possible codebooks using constructions from \cite{arya} and \cite{artyom}. The third column in the table (labeled 'Upper Bound') is the upper bound from Corollary~\ref{cor:card1} (Section~\ref{sec:bound}), which can also be found in \cite{Kashyap}.


\begin{table}
\center{\caption{Upper and Lower Bounds for single grain-error-correcting Codes}\label{table:results}
\begin{tabular}{|c|cc|c|}
	\hline
Length  & Previous                          & Current                                                                & Upper Bound   \\
              & Lower Bound                  & Lower Bound                                                       &    \\
	\hline
3                       &  4 \cite{artyom}          			&   4 \cite{artyom}                    &  4 \cite{artyom}             \\
4                       &  6 \cite{artyom}                                      &   6 \cite{artyom}                    &  6 \cite{artyom}          \\
5                       &  8 \cite{artyom}                                      &   8 \cite{artyom}                    &  8 \cite{artyom}            \\
6                       & 16 \cite{artyom}                                    &  16 \cite{artyom}                   &16 \cite{artyom}          \\
7                       & 26 \cite{artyom}                                    &   26 \cite{artyom}                  & 26 \cite{artyom}           \\
8                       & 44\cite{artyom}                                     &   44 \cite{artyom}                  & 44 \cite{artyom}       \\
9                       & 44 \cite{artyom}                                    &   {{64}}               & 112        \\
10                    &  64 \cite{arya}                                        &   {{110}}            & 204  \\
11                    &  128 \cite{arya}                                      &   {{210}}            & 372  \\   
12                    &  256 \cite{arya}                                      &   {{360}}            & 682  \\
13                    &  512 \cite{arya}                                      &   {{702}}            & 1260  \\
14                    &  1024 \cite{arya}                                    &   {{1200}}         & 2340 \\
15                    & 2176 \cite{artyom}                                     &   {{2400}}          & 4368 \\
16                    & 4096 \cite{arya}                                     &   {{4096}}           & 8190  \\
17                    & 4096 \cite{arya}                                     &   {{7712}}         & 15420 \\
18                    & 8192 \cite{arya}                                      &   {{14592}}       & 29126 \\
19                    & 16384 \cite{arya}                                  &    {{27596}}       & 55188 \\
20                    &  32768 \cite{arya}                                 &    {{52432}}         & 104856 \\
	\hline
\end{tabular}\label{table:error_patterns}}
\end{table}

\subsection{Multiple grain-error codes using the $\Gamma$ coloring}\label{subsec:manygrains}

In this subsection, multiple grain-error-correcting codes are studied. In particular, we consider an alternative interpretation of the codes from Example~\ref{ex:coloring}. Using this interpretation, we derive a lower bound on the size of a mineral code created according to Example~\ref{ex:coloring} for the case where the code $\cC_t$ is linear.



Notice that if the Hamming weight enumerator for the constituent code $\cC_t$ in Example~\ref{ex:coloring} is given, then the size of the code $\cC$ can be expressed as a function of the Hamming weight enumerator for $\cC_t$. We denote the Hamming weight enumerator of a code $\cC$ as $W_{\cC}(x,z) = \sum_{i=0}^n W_{i,n-i} z^i x^{n-i}$ where $W_{i,n-i}$ represents the number of codewords in $\cC$ whose Hamming weight is $i$. The following lemma is similar to Theorem~9 in \cite{grassl} and so the proof is omitted.  

\begin{lemma}\label{lem:cards} Let $\cC_t$ be a ternary code of length $n$ used in Example~\ref{ex:coloring} with Hamming weight enumerator
$$ \cW_{\cC_t} (x,z)= \sum_{i=0}^n W_{i,n-i} z^{i} x^{n-i}. $$
Then, the resulting mineral-error-correcting code $\cC$ (as in Example~\ref{ex:coloring}) has cardinality $|\cC| = W_{\cC_t} (2, 1)$. Prepending an additional information bit to every codeword in $\cC$ results in a grain-error-correcting code with cardinality $2 |\cC|$.
\end{lemma}

\begin{remark} Note that in general the weight enumerator for any $t$-unrestricted-error-correcting ternary code $\cC_t$ is not necessarily known. \end{remark}


 {We now summarize our results in Table~\ref{table:results2}. For each $t$, we report a triplet of values. Since for $1 < t < n$ there were no existing grain-error-correcting codebooks to compare with, we naively constructed a $t$-grain-error-correcting code by prepending an additional information bit to the start of a $t$-unrestricted-error-correcting code. As a result, the first entry in each triplet (labeled 'UEC') is the cardinality of the largest linear $t$-unrestricted-error-correcting binary code found in \cite{Magma} of length $n-1$ prepended by an additional information bit. For the second entry in each triplet (labeled 'Example~\ref{ex:coloring}*') we rely on the results from Example~\ref{ex:coloring}. Appending $\{ 00, 11\}$ to any codeword of a $t$-grain-error-correcting code of length $n$ results in a $t$-grain-error-correcting code of length-$(n+2)$; we conclude that $M(n+2,t) \geq 2M(n,t)$ \footnote{The authors thank the anonymous reviewer for pointing out this useful property.}. This observation allows us to improve the cardinalities obtained by Example~\ref{ex:coloring} in certain cases. In particular, each entry under 'Example~\ref{ex:coloring}*' is the maximum of the cardinality of a code created with Example~\ref{ex:coloring} and a code obtained by appending $\{00, 11\}$ to a shorter length code. Lastly, the third entry in each triplet (labeled 'UB') is the non-asymptotic upper bound from Theorem~\ref{th:maxcard},   rounded to an even integer, as per Lemma~\ref{lem:even}. 


In the following, we provide a variation of the codes from Example~\ref{ex:coloring} in order to {derive} an explicit lower bound on the size of codes created as in Example~\ref{ex:coloring} when $\cC_t$ is linear. 


\begin{Construction}\label{constructc} Let $r,\ell$ be positive integers where $r \leq \ell$. Let $H' = ( \bfh_1', \ldots, \bfh_{\ell}')$ be an $r \times \ell$ parity check matrix of a ternary code $\cC'$ of length $\ell$ that can correct up to $t$ unrestricted-errors (where each $\bfh_i'$ represents the $i$th column in $H'$, $1 \leq i \leq \ell$). Let $H$ be an $r \times 2\ell$ ternary matrix, $$H=(\bfh_1, \ldots, \bfh_{2\ell})=( 2\bfh_1', \bfh_1', 2\bfh_2', \bfh_2', \ldots, 2\bfh_{\ell}', \bfh_{\ell}').$$ Let $\bfa$ be an arbitrary element in $GF(3)^r$. Then,
\begin{align}\label{eq:constructccond}
\cC_{\bfa} = \{ \bfx \in GF(2)^{2\ell} : H \bfx = \bfa \},
\end{align}
where the vector operations are performed in the vector space $GF(3)^r$.
\end{Construction}

\begin{table*}
\center{\caption{Cardinalities of Grain-error-correcting Codes}\label{table:results2}
\begin{tabular}{|c|ccc|ccc|ccc|ccc|}
	\hline
& & $t=2$ & 	&  & $t=3$ & 		      &  & $t=4$ & 	&  & $t=5$ &  \\
Length	     & UEC &  {Example~\ref{ex:coloring}*} & UB 			& UEC &  {Example~\ref{ex:coloring}*} & UB			      & UEC &  {Example~\ref{ex:coloring}*} & UB	&  UEC &  {Example~\ref{ex:coloring}*} & UB \\
	\hline
11          & 16 & 68 & 168                                        & - & - &-                                                                          & - &- & -                                                    & - & - & -  \\
13          & 32 &  {136}	    & 476                                      & - & - & -                                                                         & - & - & -                                                   & - & - & - \\
15          & 128 & 312 & 1406                                  & 32 & 260 &  800                                       	  	         & - & - & -                                                     & - & - & - \\
17          & 512 & 836 & 4302                                  & 64 &  {520} & 2132               & - & - & -                                                     & - & - & - \\
19          & 1024 & 2636 & 13560                            & 256 & {1040} & 5892            & 16 &1028 & 3854               			  & - & - & - \\
21          & 4096 & 9376 & 43804                            & 1024 & 2144 & 16836                              			& 64 &  {2056}      & 9878  	     			& - & - & - \\
23          & 16384 & 35648 & 144380                      & 4096 & 4688 & 49572                              			& 128 &  {4112}    & 26100          			& 32 & 4100 & 18740 \\
25          & 32768 &  {71296}  & 483954		     & 8192 &  {9376} & 149804    & 256 & 8320 & 71018          			 & 64 &  {8200} & 46762 \\
27          & 131072 & 190912 & 1645392                 & 16384 & 20808 & 463074                       			& 1024 & 17216 & 198660    			 & 256 &  {16400} & 119626 \\
29          & 524288 & 747520 & 5662422                 & 32768 & 53460 & 1459848		     			& 2048 &  {34432} & 570038   			& 512 &  {32800} & 313846 \\
	\hline
\end{tabular}\label{table:grain-errors}}
\end{table*}


The following lemma will be useful in proving the correctness of Construction~\ref{constructc}.

\begin{lemma}\label{lem:equivalence} Let $r, \ell$ be positive integers where $r \leq \ell$ and let the matrices $H', H$ be as in Construction~\ref{constructc}. Then, for any $\bfx \in GF(2)^{2\ell}$, $H \cdot \bfx = H' \cdot \Gamma(\bfx)$. \end{lemma}
\begin{proof}
For any $\bfx=(x_1, \ldots, x_{2\ell}) \in GF(2)^{2\ell}$ we have $H \cdot \bfx = \sum_{i=1}^{2\ell} \bfh_i \cdot x_i$ where $\bfh_i \in GF(3)^r$. Consider the quantity
\begin{align}
H \cdot \bfx &= \sum_{i=1}^{2\ell} \bfh_i \cdot x_i \nonumber \\
&= \sum_{j=1, j\ \text{odd}}^{2\ell-1} (\bfh_j, \bfh_{j+1})  \cdot (x_{j}, x_{j+1})^T \nonumber \\
&=\sum_{j=1, j\ \text{odd}}^{2\ell-1} (2 \bfh_{\lceil \frac{j}{2} \rceil }', \bfh_{\lceil \frac{j}{2} \rceil }')  \cdot (x_{j}, x_{j+1})^T.\label{eq:lem10l1}
\end{align}
There are the $4$ possibilities for $(x_j, x_{j+1})$:
\begin{enumerate}
\item $(x_j, x_{j+1}) = (0, 0)$,
\item $(x_j, x_{j+1}) = (0, 1)$,
\item $(x_j, x_{j+1}) = (1, 0)$,
\item $(x_j, x_{j+1}) = (1, 1)$.
\end{enumerate}

It can be verified that in either of the $4$ cases, when $j$ is odd, we have (where $\Gamma$ is as defined in (\ref{eq:map}))

 $$ (2\bfh_{\lceil \frac{j}{2} \rceil }', \bfh_{\lceil \frac{j}{2} \rceil }')  \cdot (x_{j}, x_{j+1})^T = \bfh_{\lceil \frac{j}{2} \rceil }'  \cdot \Gamma (x_{j}, x_{j+1}). $$

Then, continuing from (\ref{eq:lem10l1}),

\begin{align*}
H \cdot \bfx &=\sum_{j=1, j\ \text{odd}}^{2\ell-1} (2 \bfh_{\lceil \frac{j}{2} \rceil}', \bfh_{\lceil \frac{j}{2} \rceil}')  \cdot (x_{j}, x_{j+1})^T \\
&= \sum_{j=1, j\ \text{odd}}^{2\ell-1} \bfh_{\lceil \frac{j}{2} \rceil }'  \cdot \Gamma (x_{j}, x_{j+1})\\
&= \sum_{k=1}^{\ell} \bfh_{k}'  \cdot \Gamma (x_{2k-1}, x_{2k}) \\
&= H' \cdot \Gamma(\bfx).
\end{align*}
\end{proof}

We now prove the correctness of Construction~\ref{constructc}.

\begin{theorem}\label{th:constructc} Suppose $\cC_{\bfa}$ is a code created according to Construction~\ref{constructc}. Then, $\cC_{\bfa}$ is a $t$-mineral-error-correcting code.
\end{theorem}
\begin{proof} Let $H'$ be a parity check matrix of dimension $r$ (where $r \leq \ell$) for the code $\cC'$ of length $\ell$ that can correct up to $t$ unrestricted-errors. For any $\bfa \in \cA$, let $\cC'_{\bfa} = \{ {\bfx} \in GF(3)^{\ell} : H' \cdot {\bfx} = {\bfa}  \}$. Notice that for any $\bfa \in \cA$, $\cC'_{\bfa}$ is a ternary $t$-unrestricted-error-correcting code. Recall from Construction~\ref{constructc} that $\cC_{\bfa}= \{ \bfx \in GF(2)^{2\ell} : H \bfx = \bfa \}$ where $H=(2 \bfh_1', \bfh_1', 2 \bfh_2', \bfh_2', \ldots, 2 \bfh_{\ell}, \bfh_{\ell}' )=(\bfh_1, \ldots, \bfh_{2\ell})$ (and each $\bfh_i, \bfh_j'$ denotes a column in $H$ or $H'$, respectively for $1 \leq i \leq 2\ell$ and $1 \leq j \leq \ell$). 

From Lemma~\ref{lem:equivalence}, for any vector $\bfx \in GF(2)^n$, $H \cdot \bfx = H' \cdot \Gamma(\bfx)$. Therefore, it follows that $H \cdot \bfx = \bfa$ if and only if $H' \cdot \Gamma(\bfx) = \bfa$. Then, we can write $\cC_{\bfa} = \{ \bfx \in GF(2)^n : \Gamma(\bfx) \in \cC_{\bfa}' \}$. Since $\cC_a'$ is a $t$-unrestricted-error-correcting code, $\cC_{\bfa}$ is a $t$-mineral-error-correcting code by Example~\ref{ex:coloring} and Theorem~\ref{th:mainconstruct}.
\end{proof}

Using the interpretation of the codes from Example~\ref{ex:coloring} provided by Construction~\ref{constructc}, we now state a simple lower bound on the size of a code created as in Example~\ref{ex:coloring}. Recall from Theorem~\ref{th:constructc}, Construction~\ref{constructc} is a special case of the codes from Example~\ref{ex:coloring}. The lower bound in Corollary~\ref{cor:roughcard} will be improved in the next section.

In the following corollary $\cA$ denotes the additive Abelian group of $GF(3)^r$.

\begin{corollary}\label{cor:roughcard} Let $\cC'$ be a $t$-unrestricted-error-correcting ternary code of length $\ell=\frac{n}{2}$ (where $n$ is even) with a parity check matrix $H'$ of dimension $r$. Then there exists an $\bfa \in \cA$, such that the code $\cC_{\bfa}$ created according to Construction~\ref{constructc} of length $n$ with the constituent code $\cC'$ satisfies $|\cC_{\bfa}| \geq \frac{2^n}{3^r}$.   \end{corollary}
\begin{proof} Notice that  each of the $2^{2\ell}$ vectors from $GF(2)^{2\ell}$ will map to exactly one code $\cC_{\bfa}$ as in (\ref{eq:constructccond}). Thus, the matrix $H$ partitions the space $GF(2)^{2\ell}$ into $|\cA|$ non-overlapping codes $\cC_{\bfa_1}, \cC_{\bfa_2}, \cC_{\bfa_3}, \ldots, \cC_{\bfa_{3^r}}$ where each $\bfa_i \in \cA$ for $1 \leq i \leq 3^r$. By the pigeonhole principle, there must exist a code with cardinality at least $\frac{2^{2\ell}}{|\cA|} = \frac{2^{n}}{3^r}$.
\end{proof}

Recall, Construction~\ref{constructc} was introduced as a tool that can be used to provide lower bounds on the sizes of mineral codes created according to the more general Construction~\ref{construct:color}. Although Lemma~\ref{lem:cards} gives a lower bound on a mineral code created according to Example~\ref{ex:coloring}, Lemma~\ref{lem:cards} has a potential drawback that it requires knowledge of the weight enumerator for a non-binary code which may not be known. The purpose of the next section is to use Construction~\ref{constructc} to produce a lower bound that holds for general $n$. The lower bound in the next section only requires the knowledge of the number of parity symbols for the non-binary constituent code. It will be demonstrated in Table~\ref{table:results3} that in many cases the resulting lower bound guarantees codebooks with strictly more codewords than the largest known binary codebooks capable of correcting of correcting a prescribed number of unrestricted-errors. In the next section, we use Fourier analysis to improve the lower bound on $\cC_{\bfa}$ from Construction~\ref{constructc}.



\section{An Improvement on the lower bounds on the cardinality of grain and mineral codes when $t \geq 2$}\label{sec:cardinalities}
In this section, we improve the lower bound from the previous section on the cardinality of a $t$-mineral-error-correcting code created according to Construction~\ref{constructc}. The approach will be similar to \cite{Mceliece}, where the cardinalities of the Constantin-Rao codes \cite{cr} were derived using discrete Fourier analysis. 

Let $\cA$ be the additive Abelian group of $GF(3)^r$. Let $\cC_{\bfa}$ denote a code created using Construction~\ref{constructc} where as before $\bfa$ is an element from $\cA$ used in the construction. Suppose further that $\cC'$ is a ternary code of length $\ell$ with an \textcolor{red}{$r \times \ell$} parity check matrix $H'$ that can correct up to $t$ unrestricted-errors where $\cC'$ is the constituent code used in Construction~\ref{constructc}. For $1 \leq i \leq \ell$, recall from the construction that $\bfh_i'$ refers to the $i$th column of $H'$ and that for $1 \leq j \leq 2\ell$, $\bfh_j$ refers to the $j$th column of $H$  where $H=(\bfh_1, \ldots, \bfh_{2\ell})=( 2\bfh_1', \bfh_1', 2\bfh_2', \bfh_2', \ldots, 2\bfh_{\ell}', \bfh_{\ell}')$.

For $\bfx = (x_1, \ldots, x_{2\ell}) \in GF(2)^{2\ell}$, consider the mapping $\gamma : GF(2)^{2\ell} \to \cA$ defined as
\begin{align}\label{eq:gamma2}
\gamma(\bfx) = H \cdot \bfx = \sum_{j=1}^{2\ell} x_j \bfh_j = \sum_{i=1}^{\ell} x_{2i} \bfh'_i + \sum_{k=1}^{\ell} 2x_{2k-1} \bfh'_k. 
\end{align}
In order to compute $|\cC_{\bfa}|$, we count the number of times each element $\bfa \in \cA$ is covered by some vector $\bfx \in GF(2)^{2\ell}$ through $\gamma$. Let $f : \cA \to \mathbb{N}$ where
\begin{align}\label{eq:fdef}
f(\bfa) = | \{ \bfx \in GF(2)^{2\ell} : \gamma(\bfx) = H \cdot {\bfx} = {\bfa} \} |.
\end{align}

We state the following claim for clarity. Recall, $M(n,t)$ refers to the maximum size of a $t$-grain-error-correcting code of length $n$.

\begin{claim}\label{cl:connect1} Let $n, \ell$ be positive integers such that $n = 2\ell + 1$. Let $\cC_{\bfa}$ be a code of length $2\ell$ created according to Construction~\ref{constructc} where $\bfa \in \cA$. Then, $|\cC_{\bfa}| = f(\bfa)$ and $M(n,t) \geq 2|\cC_{\bfa}| = 2f(\bfa)$. \end{claim}

We are now ready to derive lower bounds on the sizes of codes created from Construction~\ref{constructc} using Fourier analysis. The following lemma will be used in the proof of Theorem~\ref{th:cardstated}. Recall from Section~\ref{subsec:fourier}, for $\bfa, \bfb \in \cA$, $$  {\langle \bfa, \bfb \rangle} = \prod_{i=1}^r (\zeta_3)^{a_i b_i} = (\zeta_3)^{\sum_{i=1}^r a_i b_i} = (\zeta_3)^{{\bfa}^T \cdot {\bfb}}$$ where $\zeta_3$ is a primitive third root of unity. In the remainder, for some positive integer $k$, $(\zeta_3)^{k}$ will be written as $\zeta_3^{k}$.

In the next lemma, we make use of the following function. Let $F : \cA \times \cA \to \mathbb{C}$ where for $\bfa, \bfb \in \cA$,
$$F({\bfa}, {\bfb}) = 1 +  {\langle -{\bfa}, {\bfb} \rangle} +  {\langle -{\bfa}, 2 {\bfb} \rangle} +  {\langle -{\bfa}, {\bfb} \rangle}  {\langle -{\bfa}, {2\bfb} \rangle}. $$

\begin{lemma}\label{lem:cardtool} For any $\bfa, \bfb \in \cA$, 
\begin{equation*} F({\bfa}, {\bfb})=
\begin{cases} 4 & \mbox{if } \bfa^T \cdot \bfb = \sum_{i=1}^r a_i b_i \equiv 0 \mod 3, \mbox{ and} \\ 1 & \mbox{} \mbox{otherwise.} \end{cases}
\end{equation*}
\end{lemma}
\begin{proof}
First consider the case where $\bfa^T \cdot \bfb \equiv 0 \mod 3$. Notice that if $\bfa^T \cdot \bfb \equiv 0 \mod 3$, then $ {\langle {\bfa}, {\bfb} \rangle} = 1$. Since $\bfa^T \cdot \bfb \equiv 0 \mod 3$ we have $- \bfa^T \cdot \bfb =  - \bfa^T \cdot 2 \bfb \equiv 0 \mod 3$ and so the quantity in the Lemma is equal to $4$.

Consider the case now where $\bfa^T \cdot \bfb \not \equiv 0 \mod 3$. Recall $\zeta_3$ is a cubic root of unity and note that $ {\langle -{\bfa}, 2 {\bfb} \rangle} =  {\langle - {\bfa}, {\bfb} \rangle}^2$. Then,
\begin{align*}
&  {\langle -{\bfa}, {\bfb} \rangle} +  {\langle -{\bfa}, 2 {\bfb} \rangle} +  {\langle -{\bfa}, {\bfb} \rangle}  {\langle -{\bfa}, 2 {\bfb} \rangle} \\
&=   {\langle -{\bfa}, {\bfb} \rangle} +  {\langle -{\bfa}, {\bfb} \rangle}^2 +  {\langle -{\bfa}, {\bfb} \rangle}^3 \\
&=  \zeta_3 + \zeta_3^{2} + \zeta_3^{3} \\
&= 0, \end{align*}
and so $F({\bfa}, {\bfb}) = 1$.
\end{proof}

Given an input $\bfc \in \cA$, let $\beta : \cA \to \{0, \ldots, \ell \}$ be defined as follows

\begin{align}
\beta(\bfc) = | \{ 1 \leq i \leq \ell: \bfc^T \cdot \bfh'_i \equiv 0 \mod 3 \} |
\end{align} 
where $\bfh'_i$ refers to the $i$-th column of $H'$.

The following function will be used in the proof of Theorem~\ref{th:cardstated}. Let $I : GF(2)^{2\ell} \times \cA \to \{0, 1\}$ denote the indicator function where for $\bfx \in GF(2)^{2 \ell}$ and $\bfa \in \cA$,
\begin{equation}\label{eq:indicator1} I( {\bfx}, {\bfa} )=
\begin{cases} 1 & \mbox{if } \gamma(\bfx)={\bfa}, \mbox{ } \\ 0 & \mbox{} \mbox{otherwise} \end{cases}
\end{equation}
where $\gamma$ is as defined in (\ref{eq:gamma2}).

We are now ready for the main result of this section.
\begin{table*}
\center{\caption{Comparison of sizes of grain-error-correcting codes with the lower bound from Corollary~\ref{cor:hardcard}   \label{table:results3}}
\begin{tabular}{|c|cc|cc|cc|}
	\hline
Length  &  \multicolumn{2}{c|}{$t=2$}  & \multicolumn{2}{c|}{$t=3$} 	&  \multicolumn{2}{c|}{$t=4$}           \\
             &  {Example~\ref{ex:coloring}*} & Corollary~\ref{cor:hardcard} &  {Example~\ref{ex:coloring}*} & Corollary~\ref{cor:hardcard}   			&  {Example~\ref{ex:coloring}*} & Corollary~\ref{cor:hardcard} \\
	\hline
11          & 68 &  34                                        			    & - & -                                                                &  - & -                          \\
13          &  {136} & 46                                      		            & - & -                                                                &  - & -              \\
15          & 312 & 148                                    			   & 260 & 64                                    & - &  -                   \\
17          & 836 & 552                                    &  {520} & 86                                    & - & -   \\
19          & 2636 & 2170                                &  {1040} & 116                                & 1028 & 116  \\
21          & 9376 & 8642                                & 2144 & 356                                &  {2056} & 156    \\
23          & 35648 & 34534                            & 4688 & 1314                              &  {4112} & 208  \\
25          &  {71296} & 46044        		    	  &  {9376} & 1754                              & 8320 & 634 \\
27          & 190912 & 184130                        & 20808 & 6868                            & 17216 & 2340 \\
29          & 747520 & 736466                        & 53460 & 27324    		      &  {34432} & 3122 \\
	\hline
\end{tabular}\label{table:results3}}
\end{table*}

\begin{theorem}\label{th:cardstated} For any $\bfb \in \cA$, $f(\bfb) = \frac{1}{3^r} \sum_{\bfa \in \cA}  {\langle {\bfb}, {\bfa} \rangle} 4^{\beta({\bfa})} $.
\end{theorem}
\begin{proof} Consider $\bfc \in \cA$. As in \cite{Mceliece}, we proceed by computing the Fourier transform $\hat{f}(\bfc)$ (as defined as in Section~\ref{subsec:fourier}). First note that from (\ref{eq:indicator1}), we can write $f(\bfa) = \sum_{\bfx \in GF(2)^{2 \ell}} I({\bfx}, {\bfa})$ where $\bfa \in \cA$. We have
\begin{align*}
\hat{f}({\bfc}) &= \sum_{{\bfa} \in \cA}  {\langle {\bfc}, -{\bfa} \rangle} f({\bfa}) \\
&= \sum_{{\bfa} \in \cA}  {\langle -{\bfc}, {\bfa} \rangle} f({\bfa}) \\
&=  \sum_{{\bfa} \in \cA}  {\langle -{\bfc}, {\bfa} \rangle} \sum_{x \in GF(2)^{2 \ell}} I({\bfx}, {\bfa}) \\
&= \sum_{{\bfa} \in \cA} \sum_{\bfx \in GF(2)^{2\ell}}  {\langle -{\bfc}, {\bfa} \rangle} I( {\bfx}, {\bfa} ) \\
&= \sum_{\bfx \in GF(2)^{2\ell}} \sum_{{\bfa} \in \cA}  {\langle -{\bfc},{\bfa} \rangle}  I({\bfx}, {\bfa}).
\end{align*}
Note that for a fixed $\bfx \in GF(2)^{2\ell}$, $\sum_{{\bfa} \in \cA}  {\langle -{\bfc},{\bfa} \rangle}  I({\bfx}, {\bfa}) =  {\langle -{\bfc}, \gamma(\bfx) \rangle}$. Then,
\begin{align*}
\hat{f}({\bfc}) &=  \sum_{\bfx \in GF(2)^{2\ell}}  {\langle -{\bfc}, \gamma(\bfx) \rangle} \\
&= \sum_{\bfx \in GF(2)^{2\ell}} \langle - {{\bfc}, x_1 {\bfh}_1 + \cdots + x_{2\ell} {\bfh_{2\ell}} \rangle} \\
&=\sum_{\bfx \in GF(2)^{2\ell}}  {\langle -{\bfc}, x_1 {\bfh_1} \rangle} \cdots  {\langle -{\bfc}, x_{2\ell} {\bfh_{2\ell}} \rangle},
\end{align*}
where the last equality follows from the property that for $\bfa_1, \bfa_2, \bfa_3 \in \cA$, $ {\langle -{\bfa_1}, {\bfa_2} + {\bfa_3} \rangle} =  {\langle -{\bfa_1}, {\bfa_2} \rangle}  {\langle -{\bfa_1}, {\bfa_3} \rangle}$. 

Notice that each $x_i$ is equal to either $0$ or $1$ (where $1 \leq i \leq 2\ell$). If $x_i=0$, then clearly ${\langle -{\bfc}, x_i {\bfh_i} \rangle} = 1$. If $x_i=1$, ${\langle -{\bfc}, x_i {\bfh_i} \rangle} =  {\langle -{\bfc}, {\bfh_i} \rangle}$. As a result of this observation, we can write $\sum_{\bfx \in GF(2)^1} \langle -c, x_1 {\bfh}_1 \rangle = 1 + {\langle -{\bfc}, {\bfh}_1 \rangle}$. We now show that for $k \geq 1$, we can write 
\begin{align}\label{eq:th5res}
\sum_{\bfx \in GF(2)^{k}}  {\langle -{\bfc}, x_1 {\bfh_1} \rangle} \cdots  {\langle -{\bfc}, x_{k} {\bfh_{k}} \rangle} &=\prod_{i=1}^{k} ( 1 +  {\langle -{\bfc}, {\bfh}_i \rangle} ).
\end{align}
We will show the result by induction on $k$. The case where $k=1$ has already been proven. Suppose the result holds for all $k \leq v-1$ and consider the case where $k=v$. Let $$K = \sum_{\bfx \in GF(2)^{k-1}}  {\langle -{\bfc}, x_1 {\bfh_1} \rangle} \cdots  {\langle -{\bfc}, x_{k-1} {\bfh_{k-1}} \rangle}.$$ Similar to before, notice that if $x_{k} = 0$, then $$\sum_{\bfx \in GF(2)^{k}, x_{k}=0}  {\langle -{\bfc}, x_1 {\bfh_1} \rangle} \cdots  {\langle -{\bfc}, x_{k} {\bfh_{k}} \rangle} = K.$$ Otherwise if $x_{k}=1$, then $$\sum_{\bfx \in GF(2)^{k}, x_{k}=1}  {\langle -{\bfc}, x_1 {\bfh_1} \rangle} \cdots  {\langle -{\bfc}, x_{k} {\bfh_{k}} \rangle} = K {\langle -{\bfc}, {\bfh_{k}} \rangle}.$$ Applying the inductive hypothesis, we can write $K = \prod_{i=1}^{k-1} ( 1 +  {\langle -{\bfc}, {\bfh}_i \rangle})$. Collecting terms then gives 
$$ \sum_{\bfx \in GF(2)^{k}}  {\langle -{\bfc}, x_1 {\bfh_1} \rangle} \cdots  {\langle -{\bfc}, x_{k} {\bfh_{k}} \rangle} = K(1+{\langle -{\bfc}, {\bfh_{k}} \rangle}),$$
which after substituting the value for $K$, gives the expression in (\ref{eq:th5res}). From (\ref{eq:th5res}), we can write 
$$ \hat{f}({\bfc}) = \prod_{i=1}^{2\ell} ( 1 + \langle -{\bfc}, {\bfh}_i \rangle). $$

Let $j$ be an integer such that $1 \leq j \leq \ell$. Then from the definition of $H$ (see also (\ref{eq:gamma2})) we can write $( 1 +  {\langle -{\bfc}, {\bfh}_{2j} \rangle} ) ( 1 +  {\langle -{\bfc}, {\bfh}_{2j-1} \rangle} ) = ( 1 +  {\langle -{\bfc}, {\bfh}_j' \rangle} ) ( 1 +  {\langle -{\bfc}, 2{\bfh}_{j}' \rangle} )$. Thus, we can rewrite $\hat{f}(\bfc)$ in terms of the $\bfh_i'$ terms so that
\begin{align}
\hat{f}({\bfc}) &= \prod_{i=1}^{\ell} ( 1 +  {\langle -{\bfc}, {\bfh}_i' \rangle} +  {\langle -{\bfc}, 2 {\bfh}_i' \rangle} + \nonumber \\
& \ \ \  \ \ \ \  \ \ \  \ \ \ \ \ \   {\langle -{\bfc}, {\bfh}_i' \rangle}  {\langle -{\bfc}, {2\bfh}_i' \rangle} ) \nonumber \\
&= \prod_{i=1}^{\ell} F({\bfc}, {\bfh}_i' ) \nonumber \\ 
&= 4^{\beta( {\bfc} ) }. \nonumber
\end{align}
The equality follows from Lemma~\ref{lem:cardtool}. Recall, from Section~\ref{subsec:fourier} that the inverse Fourier transform of $\hat{f}$ is $f(\bfb) = \frac{1}{3^{r}} \sum_{{\bfa} \in \cA}  {\langle \bfa, \bfb \rangle} \hat{f}(\bfa)$. Thus, since $\hat{f}(\bfa) = 4^{\beta( {\bfa} )}$, we have that for an element $\bfb \in \cA$,
\begin{align*}
f(\bfb) &= \frac{1}{3^{r}} \sum_{{\bfa} \in \cA}  {\langle \bfa, \bfb \rangle} \hat{f}(\bfa) \\
&= \frac{1}{3^{r}} \sum_{{\bfa} \in \cA}  {\langle \bfa, \bfb \rangle} 4^{\beta({\bfa})}.
\end{align*}
\end{proof}


\begin{corollary}\label{cor:biggestz} For any $\bfb \in \cA$, $f({\bfb}) \leq f({\bf0})$.  \end{corollary}
\begin{proof} As in \cite{Mceliece}, this is because for any ${\bfa},{\bfb} \in \cA$,  $|  {\langle {\bfa}, {\bfb} \rangle} | \leq 1$. Thus, $f(\bfb) = \frac{1}{3^r} \sum_{\bfa \in \cA}  \langle {\bfa}, {\bfb} \rangle 4^{\beta({\bfa})} \leq \frac{1}{3^r} \sum_{\bfa \in \cA} 4^{\beta({\bfa})} = f({\bf0}).$ \end{proof}

Thus, choosing $\bfa={\bf0}$ in Construction~\ref{constructc} maximizes the cardinality of the resulting code. The following lemma is another consequence of Theorem~\ref{th:cardstated}.

\begin{lemma}\label{lem:improvedlower} 
For positive integers $r, \ell$ where $r \leq \ell$, $f({\bf0}) \geq \frac{4^{\ell}}{3^r} + 2 \left( \frac{4}{3} \right)^r - 2 \cdot \frac{4}{3}$. \end{lemma}
\begin{proof} From Theorem~\ref{th:cardstated}, we have that $f({\bf0}) = \frac{1}{3^r} \sum_{{\bfa} \in \cA} 4^{\beta({\bfa})} $. Clearly, $\beta({\bf0}) = \ell$ and so $f({\bf0}) = \frac{1}{3^r} \left( 4^{\ell} + \sum_{{\bfa} \in \cA, {\bfa} \neq {\bf0}}  4^{\beta( {\bfa} )} \right)$. We define the sets $\cT_0 = \{ {\bf0} \}, \cN_0 = \{ {\bf0} \},$ and $\cN'_0 = \{ {\bf0} \}$.



In the following we define the sets $\cN_j, \cN_j',$ and $\cT_j$ recursively (starting at $j=1$) where $j$ is an integer such that $1 \leq j \leq r-1$. Consider the sub-matrix $H_j'$ consisting of the first $r-j$ columns of $H'$ where $H'$ is the parity check matrix for $\cC'$ with columns $h'_i$ and $1 \leq i \leq \ell$. Let $\cN_j=\{ \bfg \in \cA : \bfg^T \cdot H_j' = \bf0 \}$. Notice that since $H_j'$ has rank at most $r-j$, $|\cN_j| \geq 3^j$. Let $\cN_j' \subseteq \cN_j$ be such that  $|\cN_j'|= 3^{j}$ where we require for $j \geq 1$ that $\cN_{j-1}' \subset \cN'_{j}$. Let $\cT_j = \cN_j' \setminus \cN_{j-1}'$. Under this setup, for any $0 \leq k < j$, $\cT_j \cap \cT_k = \emptyset$. Now, for any $\cT_j$, we have 

$$|\cT_j| = |\cN_j'| - |\cN_{j-1}'| = 3^j - 3^{j-1}.$$

Notice that for any $\bfu \in \cT_j$,  $\beta(\bfu) = | \{ 1 \leq i \leq \ell: \bfu^T \cdot \bfh'_i = {\bf0} \} | \geq r-j$. Then since the sets $\cT_0, \cT_1, \ldots, \cT_{r-1}$ are disjoint, we can use Theorem~\ref{th:cardstated} with $\bfb={\bf0}$ to obtain $f({\bf0}) = \frac{1}{3^r} \sum_{\bfa \in \cA} 4^{\beta({\bfa})} \geq \frac{1}{3^r} |\cT_0| 4^{\ell} + \frac{1}{3^r} \sum_{j=1}^{r-1} |\cT_j| 4^{r-j}$. Finally,
\begin{align*}
f({\bf0}) &\geq \frac{1}{3^r} |\cT_0| 4^{\ell} + \frac{1}{3^r} \sum_{j=1}^{r-1} |\cT_j| 4^{r-j} \\
&\geq \frac{1}{3^r}4^{\ell} + \frac{1}{3^r} \sum_{j=1}^{r-1} (3^j - 3^{j-1}) 4^{r-j} \\
&=\frac{1}{3^r} 4^{\ell} + \frac{2 \cdot 4^{r-1}}{3^r} \sum_{j=0}^{r-2} \left( \frac{3}{4} \right)^j \\
&=\frac{1}{3^r} 4^{\ell} + 2 \left( \frac{4}{3} \right)^r - 2 \cdot \frac{4}{3},
\end{align*}
and therefore the proof is complete. \end{proof}

We summarize the result from Lemma~\ref{lem:improvedlower} with the following corollary.

\begin{corollary}\label{cor:hardcard} Let $\cC'$ be a $t$-unrestricted-error-correcting ternary code of length $\ell= \frac{n}{2}$ (where $n$ is an even integer) with a parity check matrix $H'$ of dimension $r$. For $\bfa \in \cA$, let $\cC_{\bfa}$ be a code created according to Construction~\ref{constructc} of length $n$ with the constituent code $\cC'$. Then for any $\bfa \in \cA$, $|\cC_{\bfa} | \leq | \cC_{\bf0} |$ and $|\cC_{\bf0}| \geq \lceil \frac{2^n}{3^r} + 2 \left( \frac{4}{3} \right)^r - \frac{8}{3} \rceil.$ \end{corollary}
\begin{proof} From Claim~\ref{cl:connect1}, $| \cC_a | = f(\bfa)$. Using Corollary~\ref{cor:biggestz}, we have that for any $\bfa \in \cA$, $| \cC_{\bfa} | = f(\bfa) \leq f({\bf0}) = | \cC_{\bf0} |$. Combining Claim~\ref{cl:connect1} and Lemma~\ref{lem:improvedlower} gives that $|\cC_{\bf0}| = f({\bf0}) \geq \frac{4^{\ell}}{3^r} + 2 \left( \frac{4}{3} \right)^r - 2 \cdot \frac{4}{3}$.
\end{proof}

Thus, the previous corollary improved upon Corollary~\ref{cor:roughcard} where it was shown that for some $\bfa \in \cA$, $|\cC_{\bfa}| \geq \frac{2^n}{3^r}$. For the case of $t=2,3,4$, we compared our lower bound with the cardinality of the $t$-grain-error-correcting codes from Table~\ref{table:results2}. Each entry in Table~\ref{table:results3} contains two numbers delimited by a '/'. The first number is the cardinality of a $t$-grain-error-correcting code  {from Table~\ref{table:results2} (under the sub-column Example~\ref{ex:coloring}*)} and the second number is twice the lower bound from Corollary~\ref{cor:hardcard} (since the lower bound from Corollary~\ref{cor:hardcard} applies to a mineral code and not a grain code). It can be seen in Table~\ref{table:results3} that the difference between the bound from Corollary~\ref{cor:hardcard} and the size of the codes from Table~\ref{table:results2} is small for the $t=2$ case. 


In the next section, we return to the problem of constructing single mineral codes.

\section{General Single-Grain and Single-Mineral Codes from Construction~\ref{construct:color}}\label{sec:improvedsingle}

  {In this section, we consider single-mineral codes derived from more general colorings according to Construction~\ref{construct:color}. In Section~\ref{subsec:prelims}, we investigate a sufficient condition for codes created with Construction~\ref{construct:color} to produce large single-mineral codes. In Section~\ref{subsec:newscheme}, we describe a coloring that was found using a computerized search; for code lengths $48$ and $342$ this coloring produces new codes with large cardinalities (i.e., larger than using the alternative group codes to construct single mineral-codes).}

Recall from Construction~\ref{construct:color} in Section~\ref{subsec:improvedgrains} that the construction for a $t$-mineral-error-correcting code $\cC$ relied on two key ingredients:
\begin{enumerate}
\item a mapping $\Phi_{t,m}$ from $GF(2)^m$ to $p$ color classes (where $p$ is a prime), and
\item a $t$-unrestricted-error-correcting code $\cC_t$ over $GF(p)$.
\end{enumerate}

The basic idea behind Construction~\ref{construct:color} was to use $\Phi_{t,m}$ to map the color classes of a $p$-coloring onto the symbols of the non-binary code $\cC_t$. Notice in this section, we restrict ourselves to the case where $\Phi_{t,m}$ maps to $p$ color classes where $p$ is a prime.

Thus far, we have considered code constructions for mineral codes using Construction~\ref{construct:color} with the map $\Phi_{t,m} = \Gamma$, where $\Gamma$ is given by (\ref{eq:map}). Therefore, if Construction~\ref{construct:color} is used to create mineral codes, there are two possible directions to investigate:
\begin{enumerate}
\item discover new mappings $\Phi_{t,m}$ for $m \geq 2$, and
\item  investigate new constructions for the code $\cC_t$ that, when used in conjunction with some $\Phi_{t,m}$, result in codes with large cardinalities.
\end{enumerate}

 In this section, we focus on the first direction for the case where $t=1$, and the code $\cC_1$ is a single unrestricted-error-correcting code taken to be a Hamming code. The second item highlights a potential area of future work which we will discuss briefly in the next section.
 

\subsection{A sufficient condition for Construction~\ref{construct:color} to produce large codes}\label{subsec:prelims}
Suppose that a single-mineral code $\cC$ of length $mn$ is created according to Construction~\ref{construct:color}. Suppose that the $p$-coloring $\Phi_{1,m}$ is such that $p=m+1$ where $p$ is an odd prime and $\cC_1$ is a perfect non-binary single unrestricted-error-correcting code over $GF(p)$ of length $n$. We show that there exists a mineral code $\cC$ of length $mn$ whose cardinality is at least $\frac{2^{mn}}{mn + 1}$. Motivated by this observation, in Section~\ref{subsec:newscheme}, we consider using different coloring schemes (i.e., where $\Phi_{1,m} \neq \Gamma$) in conjunction with a perfect single-unrestricted-error correcting code. We first begin by reviewing some notation that was used in Section~\ref{subsec:improvedgrains}.

As in Section~\ref{subsec:improvedgrains}, let $\cG_{t,m}=(V,E)$ denote a simple graph where $V=GF(2)^m$, and for any $\bfx, \bfy \in V$  (where $\bfx \neq \bfy$) $(\bfx, \bfy) \in E$ if $\cB_{t,M}(\bfx) \cap \cB_{t,M}(\bfy) \neq \emptyset$. Recall that $\Phi_{t,m} : GF(2)^m \to GF(p)$ is a $p$-coloring if it assigns different elements of $GF(p)$ to adjacent vertices. From Section~\ref{sec:confusabilitygraph}, $\chi(\cG_{t,m})$ is the smallest $p$ for which a $p$-coloring is possible. Recall, the size of the largest clique in a graph $\cG$ is denoted $\varsigma(\cG)$.

The following claim will be used in the proof of Lemma~\ref{lem:simplegraph}.

\begin{claim}\label{cl:clique} For any $m \geq 2$, $\varsigma(\cG_{1,m}) \geq m+1$.  \end{claim}
\begin{proof} Let $\cS=\{ \bfx \in GF(2)^m : wt(\bfx) \leq 1 \}$. Since for any $\bfx \in \cS$, $\cB_{1,m}(\bfx)$ contains the all-zeros vector, it follows that $\cS$ is a clique in $\cG_{1,m}$. Since $|\cS|=m+1$, the result follows. 
\end{proof}

\begin{lemma}\label{lem:simplegraph}  Let $m$ be a positive integer. Then, $\chi(\cG_{1,m}) = m+1.$  \end{lemma}
\begin{proof} We first show that $\chi(\cG_{1,m}) \leq m+1$. Suppose $\cA$ is an Abelian group. Let $a \in \cA$ and consider a single-grain code $\cC^{\cA}_{a}$ of length $|\cA| = m+1$ created using Construction~\ref{grain-construct}. Let $\tilde{\cC}^{\cA}_a$ be the group code of length $m$ that is the result of shortening the codewords in $\cC^{\cA}_a$ on the first bit. From Corollary~\ref{cor:singlemin}, $\tilde{\cC}^{\cA}_a$ is a single-mineral code. Assign to every $\bfx \in \tilde{\cC}^{\cA}_a$ the same number from $\{0, 1, \ldots, m \}$. Repeating this process for every value of $a \in \cA$ (and using a different number for different values of $a$), results in an $(m+1)$-coloring on the graph $\chi(\cG_{1,m})$ since there are $|\cA|=m+1$ choices for $a$.

Recall from Section~\ref{sec:confusabilitygraph} that $\chi(\cG_{1,m}) \geq \varsigma(\cG_{1.m})$ where $\varsigma(\cG_{1,m})$ is the maximum size of any clique in the graph $\cG_{1,m}$. From Claim~\ref{cl:clique}, we have $\chi(\cG_{1,m}) \geq \varsigma(\cG_{1,m}) \geq m+1$ and so $\chi(\cG_{1,m}) = m + 1$.  \end{proof}

The following theorem is similar to Corollary~\ref{cor:roughcard}.

\begin{theorem} Let $p$ be a prime number and $r$ a positive integer where $n = \frac{p^r - 1}{p-1}$ and $m=p-1$. Then there exists a single-mineral code $\cC$ of length $mn$ where $|\cC| \geq \frac{2^{mn}}{mn + 1}$ from Construction~\ref{construct:color}. \end{theorem}
\begin{proof} Let $\cC_1$ be the constituent non-binary single-unrestricted-error-correcting code from Construction~\ref{construct:color} (where $\cC_1$ is represented by $\cC_t$ in the statement of the construction). Assume also that $\cC_1$ is of length $n$ with parity check matrix $H'$ of dimension $r$ and that $\cC_1$ is perfect. For $\bfa \in \cA= GF(p)^r$, let $\cC_{\bfa}' = \{ \bfx' \in GF(p)^{n} : H' \cdot \bfx' = \bfa  \}$. Notice that since $\cC_1$ is a perfect single unrestricted-error-correcting code then $\cC_{\bfa}'$ is also a perfect single unrestricted-error-correcting code. Thus, we can apply Construction~\ref{construct:color} to obtain a single-mineral code $\cC_{\bfa}$ where
$$ \cC_{\bfa} = \{ \bfx \in GF(2)^{mn} : \Phi_{1,m}(\bfx) \in \cC_{\bfa}' \}. $$
Since $\Phi_{1,m}$ maps every element in $GF(2)^{mn}$ to exactly one non-binary vector of length $n$, it follows that every $\bfx \in GF(2)^{mn}$ belongs to exactly one $\cC_{\bfa}$, and so the codes $\cC_{\bfa_1}, \cC_{\bfa_2}, \ldots, \cC_{\bfa_{p^r}}$ partition the space $GF(2)^{mn}$ into $p^r$ non-overlapping sets. By the pigeonhole principle, there exists a $\bfb \in \cA$, where $|\cC_{\bfb}| \geq \frac{2^{mn}}{p^r} = \frac{2^{mn}}{mn + 1}$.
\end{proof}

\subsection{A new coloring scheme}\label{subsec:newscheme}

In this section, we report on the results of using Construction~\ref{construct:color} with a new map that was identified using a computerized search. As before, we denote the color classes as $A_0$, $A_1$, \ldots, $A_{p-1}$ for the $p$-coloring $\Phi_{t,m}$ on $\cG_{t,m}$ where $A_0 \cup A_1 \cup \cdots \cup A_{p-1} = GF(2)^m$. By this setup, we assume 
\begin{enumerate}
\item $\forall j \in \{0, \ldots, p-1 \}$, $A_j \subseteq GF(2)^m$ ,
\item for any $i,j \in \{0, \ldots, p-1 \}$ where $i \neq j$ $A_i \cap A_j = \emptyset$,
\item $|A_0| \geq |A_1| \geq \hdots \geq |A_{p-1}|$.
\end{enumerate}

In this subsection, we make use of the following notation. Suppose a code $\cC$ is a $t$-mineral-error-correcting code created according to Construction~\ref{construct:color} given by
\begin{enumerate}
\item a set of $p$ color classes $D=\{ A_0, A_1, \ldots, A_{p-1} \}$ for a $p$-coloring on $\cG_{t,m}$ where $p$ is a prime,
\item the mapping $\Phi_{t,m}$ which maps vectors from $GF(2)^m$ into the symbols $\{0, 1, \ldots, p-1\}$, 
\item $\cC_t$ where $\cC_t$ is a $t$-unrestricted-error-correcting code over $GF(p)$. 
\end{enumerate}
We denote the mineral code $\cC$ as $\cC(\Phi_{t,m}, \cC_t)$. Under this setup, the map $\Phi_{t,m}$ always maps elements from the same color class to the same symbol. Furthermore, for $0 \leq i \leq p-1$ and $\bfv \in A_i$, $\Phi_{t,m}(\bfv) = i$.

In the following, we describe the color classes from a $7$-coloring on $\cG_{1,6}$ that was located with the aid of a computer search. The vectors from $GF(2)^m$ are enumerated by their decimal representation. For example, the vector $\bfx = (x_1, x_2, x_3, x_4, x_5, x_6) = (1, 0, 1, 1, 0, 0)$ corresponds to the number $13$ since $\sum_{i=1}^6 2^{i-1} x_i = 13$ in this representation. The color classes are the following:
\\

\subsubsection{Color class $A_0$} $\{ 0, 3, 12, 15, 21, 24, 36, 43, 49, 54, 61 \}$
\subsubsection{Color class $A_1$} $\{ 1, 6, 13, 18, 25, 30, 37, 40, 59 \}$.
\subsubsection{Color class $A_2$} $\{ 4, 7, 9, 19, 31, 34, 46, 52, 57 \}$
\subsubsection{Color class $A_3$} $\{ 8, 11, 20, 23, 33, 38, 45, 50, 62 \}$
\subsubsection{Color class $A_4$} $\{ 10, 16, 22, 28, 35, 41, 47, 53, 58 \}$
\subsubsection{Color class $A_5$} $\{ 17, 26, 29, 32, 39, 44, 51, 56, 63 \}$
\subsubsection{Color class $A_6$} $\{ 2, 5, 14, 27, 42, 48, 55, 60 \}.$\\

Notice that $|A_0|=11, |A_1| = 9,  |A_2| = 9,  |A_3| = 9,  |A_4| = 9,  |A_5| = 9$, and $|A_6|=8$. Recall that if the group-code partition was used then the sizes of the color classes are $10,9,9,9,9,9,9$ so that the size of the largest color class has increased by $1$ given the new set of color classes.

Using a non-binary perfect code over $GF(7)$ of length $8$ with the coloring scheme specified in this section,  the resulting length-$48$ binary code has $16192$ more codewords than a group code defined over $\mathbb{Z}_7 \times \mathbb{Z}_7.$ Using a non-binary perfect code over $GF(7)$ of length $57$ with the coloring scheme described in this section results in a binary code of length $342$ with approximately $7.1401e34$ more codewords than a group code defined over $\mathbb{Z}_7 \times \mathbb{Z}_7 \times \mathbb{Z}_7$. The parity check matrices for the single-unrestricted-error-correcting codes of length $8$ and of length $57$ over $GF(7)$ were taken from \cite{Magma}. 



\section{Conclusion and Future Work}\label{sec:conclude}

In this work, new bounds and constructions were derived for grain-error-correcting codes where the lengths of the grains were at most two. We considered a new approach to constructing codes that correct grain-errors and using this approach, we improved upon the constructions in \cite{arya} and \cite{artyom}. 

There are many directions for future work:
\begin{enumerate}
\item Development of new coloring schemes and codes to use with Construction~\ref{construct:color}.
\item Constructions of codes that correct multiple non-overlapping grain-errors.
\item Constructions and bounds for codes capable of correcting grain-errors where the length of the grain is greater than two.
\item Constructions and bounds for codes that correct bursts of grain-errors.
\end{enumerate}

The largest single-grain codes for $9 \leq n \leq 15$ listed in Table~\ref{table:error_patterns} were the result of using Construction~\ref{construct:color} with non-linear codes over $GF(3)$. It seems promising that potentially larger single-grain codes may be possible using non-linear codes and coloring schemes in conjunction with Construction~\ref{construct:color} for longer code lengths.

There is clearly a strong connection between codes capable of correcting bursts of unidirectional errors and codes correcting grain-errors (where the length of the grain is longer than two). Constructing grain codes that are larger than the unidirectional codes from \cite{bose} could be of future research interest.

Finally, we note that Construction~\ref{construct:color} may be applicable to the construction of new asymmetric error-correcting codes for the Z-channel. In fact, when $\Phi_{t,m}=\Gamma$ and $\cC_1$ (where $\cC_1$ represents the code $\cC_t$ in Construction~\ref{construct:color})
is a single unrestricted-error-correcting ternary code, Construction~\ref{construct:color} is identical to the single asymmetric error-correcting code (from the ternary construction) described in \cite{grassl}. Given new colorings (i.e., where $\Phi_{1,m} \neq \Gamma$) and new ternary codes for $\cC_1$, it may be possible to construct new codes with large cardinalities for the Z-channel.

\section*{Acknowledgements}

The authors would like to thank Artyom Sharov for his helpful discussions about code cardinalities as well as Associate Editor Prof. Navin Kashyap and the anonymous reviewers for providing detailed and expert comments. This research was supported in part by SMART scholarship, ISEF Foundation, and NSF grants CCF-1029030 and CCF-1150212.

\appendices

\section{Proofs of Claims and Lemmas From Section~\ref{sec:bound}}\label{app:m3}
\subsection{Details for $M(n,2)$}
In this section, we derive the results stated in both Claims~\ref{cl:2} and~\ref{cl:3}. We begin with Claim~\ref{cl:2}, which is simpler to derive than Claim~\ref{cl:3}.
\begin{repclaim}{cl:2}
For $n \geq 2$, $$\sum_{k=2}^n \frac{1}{k+1} \nchoosek{n}{k} = \frac{1}{n+1} (2^{n+1} - 2 - \frac{3n}{2} - \frac{n^2}{2}).$$
\end{repclaim}
\begin{proof}
This identity follows from the following derivations:
\begin{align*}
& \sum_{k=2}^n \frac{1}{k+1} \nchoosek{n}{k} = \sum_{k=0}^n \frac{1}{k+1} \nchoosek{n}{k} - 1 - n/2 & \\
= & \sum_{k=0}^n \frac{1}{n+1} \nchoosek{n+1}{k+1} - 1 - n/2 = \frac{2^{n+1}-1}{n+1}- 1 - n/2 & \\
= & \frac{1}{n+1} (2^{n+1} - 2 - \frac{3n}{2} - \frac{n^2}{2}). &
\end{align*}
\end{proof}
The next lemma will be useful to obtain the expression for Claim~\ref{cl:3}.

\begin{lemma}\label{lem:ubc} For integers $n,c$ where $0 \leq c < n $,
$$ \sum_{k=c+1}^n \frac{1}{k-c} \nchoosek{n}{k} =  \sum_{j=1}^{n-c} \frac{ \nchoosek{n-j}{c}}{j} \cdot \left( 2^j - 1 \right).$$
 \end{lemma}
 \begin{proof}
  {From Volume 4, Equation (6.97) in \cite{gould}, we have the following identity
 $$ \sum_{\ell=\alpha}^{m-\beta} \nchoosek{\ell}{\alpha} \nchoosek{m-\ell}{\beta} = \nchoosek{m+1}{\alpha + \beta + 1}, $$
 where $m - \beta \geq \alpha$. Letting $k = \alpha + 1$, $\beta = c$, and $n = m+1$, we can rewrite the identity as $\nchoosek{n}{k+c} = \sum_{\ell=k-1}^{n-c-1} \nchoosek{\ell}{k-1} \nchoosek{n-1-\ell}{c}$, which after reindexing, gives
 $$ \sum_{j = k}^{n-c} \nchoosek{j-1}{k-1} \nchoosek{n-j}{c} = \nchoosek{n}{k+c}. $$
 Using this identity along with a change of variable, we can rewrite the expression in the lemma as
 }
 \begin{align*}
 \sum_{k=c+1}^n \frac{1}{k-c} \nchoosek{n}{k} &= \sum_{k=1}^{n-c} \frac{1}{k} \cdot \sum_{j=k}^{n-c} \nchoosek{j-1}{k-1} \nchoosek{n-j}{c} \\
 &=\sum_{k=1}^{n-c} \sum_{j=k}^{n-c} \frac{1}{j} \nchoosek{j}{k} \nchoosek{n-j}{c}.
 \end{align*}
 Reversing the order of the summations gives
 \begin{align*}
 \sum_{k=1}^{n-c} \sum_{j=k}^{n-c} \frac{1}{j} \nchoosek{j}{k} \nchoosek{n-j}{c} &= \sum_{j=1}^{n-c} \frac{\nchoosek{n-j}{c}}{j} \sum_{k=1}^j \nchoosek{j}{k} \\
 &=  \sum_{j=1}^{n-c} \frac{\nchoosek{n-j}{c}}{j} (2^j - 1),
 \end{align*}
as desired.
 \end{proof}
 
 As a consequence of the previous lemma we have the following corollary.
 
\begin{corollary}\label{cor:k} For $n \geq 2$, $\sum_{k=1}^n \frac{1}{k} \nchoosek{n}{k} = \sum_{j=1}^n \frac{1}{j} (2^j - 1).$ \end{corollary}
 
We require one more lemma before proceeding to the proof of Claim~\ref{cl:3}. 
 
\begin{lemma}\label{lem:improved}
For $n\geq 17$, $$\sum_{k=1}^n\frac{2^k-1}{k}\leq \frac{2^{n+1}}{n-1-\frac{2}{n-5}+\frac{1}{n^2}}.$$
\end{lemma}
\begin{proof}
The proof follows by induction for $n\geq 17$.  For $n=17$, the value on the left hand side is equal 16552.47, while the value of the right hand side is equal 16552.85. Now, assume the inequality holds for some $n\geq 17$, and we will show its validity for $n+1$. Hence, we need to show that
$$\sum_{k=1}^{n+1}\frac{2^k-1}{k}\leq \frac{2^{n+2}}{n-\frac{2}{n-4}+\frac{1}{(n+1)^2}}.$$
According to the induction assumption, it is enough to show that
$$\frac{2^{n+1}}{n-1-\frac{2}{n-5}+\frac{1}{n^2}} + \frac{2^{n+1}-1}{n+1} \leq \frac{2^{n+2}}{n-\frac{2}{n-4}+\frac{1}{(n+1)^2}},$$ or
\begin{align*}\label{eq:lem171}
\frac{1}{n-1-\frac{2}{n-5}+\frac{1}{n^2}} + \frac{1}{n+1} \leq \frac{2}{n-\frac{2}{n-4} + \frac{1} {(n+1)^2}}.
\end{align*}
 {
The last inequality is equivalent to showing that 
$$n^5-11n^4-101n^3-6n^2+73n+10\geq 0.$$
The value of the term on the left hand side of this inequality for $n=17$ is positive. Furthermore, it is possible to verify that the derivative of this term for values greater than or equal to $17$ is positive and hence we conclude that this inequality holds for $n\geq 17$.
}
\end{proof}
Claim~\ref{cl:3} now follows from Corollary~\ref{cor:k} and Lemma~\ref{lem:improved}.

\begin{repclaim}{cl:3}
For $n \geq 17$,
$$\sum_{k=1}^n \frac{1}{k} \nchoosek{n}{k} \leq \frac{2^{n+1}}{n-1-\frac{2}{n-5}+\frac{1}{n^2}}.$$
\end{repclaim}
%

\subsection{Details for $M(n,3)$}

We now proceed to derive similar results for $M(n,3)$. We first note the following corollary which follows from Lemma~\ref{lem:ubc}.

\begin{corollary}\label{cor:k-1} For $n \geq 2$, 
 $$\sum_{k=2}^n \frac{1}{k-1} \nchoosek{n}{k} = n \sum_{k=1}^{n-1} \frac{2^k-1}{k} - 2^n + n + 1.$$ \end{corollary}
 \begin{proof}
 From Lemma~\ref{lem:ubc}, we have
 \begin{align*}
 \sum_{k=2}^n \frac{1}{k-1} \nchoosek{n}{k} &= \sum_{j=1}^{n-1} \frac{n-j}{j} \cdot (2^j - 1) \\
 &= n \sum_{j=1}^{n-1} \frac{2^j - 1}{j} - \sum_{j=1}^{n-1} (2^j - 1) \\
 &= n \sum_{j=1}^{n-1} \frac{2^j - 1}{j} - \left( (2^n-1) - 1 \right) + n-1,
 \end{align*}
 which simplifies to the expression in the corollary.
 \end{proof}

\begin{replemma}{lem:ub3}
For $n \geq 24$, $$M(n,3) \leq 2 \Big\lfloor 2^n\left( \frac{\frac{36n}{n-7} + 18 - \frac{3n(n-1)}{(n-2)^2}  + \frac{12}{n-7}}{2n(n-1)(n-3-\frac{2}{n-7} + \frac{1}{(n-2)^2}} \right) \Big\rfloor.$$
\end{replemma}
\begin{proof}
From Theorem~\ref{th:maxcard} we have
\begin{small}
\begin{align*}
& M(n,3) \leq 2 \sum_{k=0}^{n-1} \nchoosek{n-1}{k} \frac{1}{\sum_{j=0}^{\min\{3,k\}} \nchoosek{k}{j}}  \\
& = 2 + n-1+\frac{1}{2}\cdot\binom{n-1}{2}+  2\sum_{k=3}^{n-1}  \nchoosek{n-1}{k} \frac{1}{1+k+\binom{k}{2}+\binom{k}{3}} \\
& \leq \frac{n^2+n+6}{4} +2\sum_{k=3}^{n-1}  \nchoosek{n-1}{k} \frac{1}{\binom{k}{2}+\binom{k}{3}} \\
& = \frac{n^2+n+6}{4} +12\sum_{k=3}^{n-1} \nchoosek{n-1}{k} \left( \frac{1/2}{k-1}-\frac{1}{k}+\frac{1/2}{k+1} \right) \\
& = \frac{n^2+n+6}{4} +6\sum_{k=3}^{n-1} \nchoosek{n-1}{k} \frac{1}{k-1} & \\
& -12\sum_{k=3}^{n-1} \nchoosek{n-1}{k}\frac{1}{k}+6\sum_{k=3}^{n-1} \nchoosek{n-1}{k}\frac{1}{k+1}.
\end{align*}
\end{small}
According to Corollary~\ref{cor:k-1},
\begin{small}
$$\sum_{k=3}^{n-1}\binom{n-1}{k}\frac{1}{k-1} = (n-1)\sum_{k=1}^{n-2}\frac{2^k-1}{k} - 2^{n-1}-\frac{n^2-5n+2}{2}.$$
\end{small}
By Corollary~\ref{cor:k},
\begin{small}
$$\sum_{k=3}^{n-1}\binom{n-1}{k}\frac{1}{k} = \sum_{k=1}^{n-1}\frac{2^k-1}{k} - \frac{n^2+n-2}{4},$$
\end{small}
and by Claim~\ref{cl:2},
\begin{small}
\begin{align*}
&\sum_{k=3}^{n-1}\binom{n-1}{k}\frac{1}{k+1} = \frac{1}{n}\left(2^n-2-\frac{3(n-1)}{2}-\frac{(n-1)^2}{2}\right)\\ 
&\ \ \  \ \ -\frac{n^2-3n+2}{6}.
\end{align*}
\end{small}
All together we get that
\begin{small}
\begin{align*}
& M(n,3) \leq \frac{n^2+n+6}{4} & \\
& + 6\left(  (n-1)\sum_{k=1}^{n-2}\frac{2^k-1}{k} - 2^{n-1}-\frac{n^2-5n+2}{2} \right) & \\
& -12\left(  \sum_{k=1}^{n-1}\frac{2^k-1}{k} - \frac{n^2+n-2}{4} \right) & \\
& +6\left( \frac{1}{n}\left(2^n-2-\frac{3(n-1)}{2}-\frac{(n-1)^2}{2}\right)-\frac{n^2-3n+2}{6} \right) & \\
& = -\frac{3n^2}{4}+\frac{73n}{4}-\frac{31}{2}-\frac{6}{n} +6(n-1)\sum_{k=1}^{n-2}\frac{2^k-1}{k}-3\cdot 2^n &\\
& -12\sum_{k=1}^{n-1}\frac{2^k-1}{k} +\frac{6\cdot 2^n}{n} & \\
& = -\frac{3n^2}{4}+\frac{73n}{4}-\frac{31}{2}-\frac{6}{n} +6(n-3)\sum_{k=1}^{n-2}\frac{2^k-1}{k}-3\cdot 2^n &\\
& -12\cdot \frac{2^{n-1}-1}{n-1} +\frac{6\cdot 2^n}{n} & \\
& \leq 6(n-3)\sum_{k=1}^{n-2}\frac{2^k-1}{k}-3\cdot 2^n -\frac{6\cdot 2^n}{n-1}+\frac{6\cdot 2^n}{n} & \\
& = 6(n-3)\sum_{k=1}^{n-2}\frac{2^k-1}{k}-3\cdot 2^n -\frac{6\cdot 2^n}{n(n-1)}&
\end{align*}
\end{small} where the inequality holds for $n \geq 24$. Finally, according to Lemma~\ref{lem:improved} we get
\begin{small}
\begin{align*}
& M(n,3) \leq 6(n-3)\frac{2^{n-1}}{n-3-\frac{2}{n-7}+\frac{1}{(n-2)^2}} -3\cdot 2^n -\frac{6\cdot 2^n}{n(n-1)} & \\
& = 2^n\left(\frac{3n-9}{n-3-\frac{2}{n-7}+\frac{1}{(n-2)^2}}-3-\frac{6}{n(n-1)}\right) & \\
& = 2^n\left(\frac{\frac{6}{n-7}-\frac{3}{(n-2)^2}}{n-3-\frac{2}{n-7}+\frac{1}{(n-2)^2}}-\frac{6}{n(n-1)}\right) & \\
&= 2^n \left( \frac{ \frac{6(n(n-1))}{n-7} -6n - \frac{3n(n-1)}{(n-2)^2} + 18 + \frac{12}{n-7} - \frac{6}{(n-2)^2}}{n(n-1)(n-3-\frac{2}{n-7} + \frac{1}{(n-2)^2})} \right) & \\
&\leq 2^n \left( \frac{ \frac{6(n(n-1))}{n-7} -6n + 18 - \frac{3n(n-1)}{(n-2)^2} + \frac{12}{n-7} }{n(n-1)(n-3-\frac{2}{n-7} + \frac{1}{(n-2)^2})} \right) & \\
& = 2^n\left( \frac{\frac{36n}{n-7} + 18 - \frac{3n(n-1)}{(n-2)^2} + \frac{12}{n-7}}{n(n-1)(n-3-\frac{2}{n-7} + \frac{1}{(n-2)^2}} \right). &
\end{align*}
\end{small}

From Lemma~\ref{lem:even}, $M(n,3)$ must be an even integer and so $M(n,3) \leq 2 \Big\lfloor 2^n\left( \frac{\frac{36n}{n-7} + 18 - \frac{3n(n-1)}{(n-2)^2}  + \frac{12}{n-7}}{2n(n-1)(n-3-\frac{2}{n-7} + \frac{1}{(n-2)^2}} \right) \Big\rfloor $.

\end{proof}

\end{document}